\documentclass[11pt]{article}

\usepackage{amsmath, latexsym, amsfonts, amssymb, amsthm,MnSymbol,stmaryrd}
\usepackage{graphicx, xcolor,hyperref,dsfont,epsfig,caption,wrapfig,subfig,pifont,tikz}
\usepackage{datetime}
\usepackage{movie15}
\usepackage{comment}
\hypersetup{
    linktoc=page,
    linkcolor=red,          
    citecolor=blue,        
    filecolor=blue,      
    urlcolor=cyan,
    colorlinks=true           
}
\newcommand\smallO{
  \mathchoice
    {{\scriptstyle\mathcal{O}}}
    {{\scriptstyle\mathcal{O}}}
    {{\scriptscriptstyle\mathcal{O}}}
    {\scalebox{.7}{$\scriptscriptstyle\mathcal{O}$}}
  }

\newcommand{\nocontentsline}[3]{}
\newcommand{\tocless}[2]{\bgroup\let\addcontentsline=\nocontentsline#1{#2}\egroup}
\numberwithin{equation}{section}

\newtheorem{theorem}{Theorem}[section]
\newtheorem{lemma}[theorem]{Lemma}
\newtheorem{proposition}[theorem]{Proposition}
\newtheorem{corollary}[theorem]{Corollary}

\newtheorem{remark}[theorem]{Remark}
\newtheorem{definition}[theorem]{Definition}

\theoremstyle{definition}

\renewcommand{\tilde}{\widetilde}          
\DeclareMathSymbol{\leqslant}{\mathalpha}{AMSa}{"36} 
\DeclareMathSymbol{\geqslant}{\mathalpha}{AMSa}{"3E} 
\DeclareMathSymbol{\eset}{\mathalpha}{AMSb}{"3F}     
\renewcommand{\leq}{\;\leqslant\;}                   
\renewcommand{\geq}{\;\geqslant\;}                   

\newcommand{\C}{\mathbb{C}}

\newcommand{\R}{\mathbb{R}}
\newcommand{\Z}{\mathbb{Z}}
\newcommand{\N}{\mathbb{N}}

\newcommand{\Q}{\mathbb{Q}}

\newcommand{\E}{\mathbb{E}}
\renewcommand{\P}{\mathbb{P}}

\usepackage{xparse}

\def\eps{\varepsilon}

\def\T{\mathbb{T}}

\def\bi{\begin{itemize}}
\def\ei{\end{itemize}}

\def\bnum{\begin{enumerate}}
\def\enum{\end{enumerate}}

\def\<#1{\langle #1 \rangle}

\newcommand{\ab}{|}

\newcommand{\les}{<}

\newcommand{\grea}{>}
\newcommand{\sd}{S_{\gamma,\eps}}
\newcommand{\Pb}{\P}
\newcommand{\opn}{\operatorname}
\newcommand{\tor}{{\T_1^2}}
\newcommand{\toR}{{\T_R^2}}

\newcommand{\PN}{ {P_{< N}} }
\newcommand{\PNN}{ {P_{>N}} }
\newcommand{\PR}{ {P_{< R}}}
\newcommand{\PRR}{ { P_{>R}} }
\newcommand{\F}[1]{ { \hat{#1} } }

\begin{document}

\title{Small deviations of Gaussian multiplicative chaos and the free energy of the two-dimensional massless Sinh--Gordon model 
}

\author{
Nikolay Barashkov\thanks{Max Planck Institute for Mathematics in the Sciences, Leipzig, Germany, \textit{nikolay.barashkov@mis.mpg.de}}\:,
Joona Oikarinen\thanks{Department of Mathematics and Systems Analysis, Aalto University, Espoo, Finland, \textit{joona.oikarinen@aalto.fi}}
and Mo Dick Wong\thanks{Department of Mathematical Sciences, Durham University, United Kingdom, \textit{mo-dick.wong@durham.ac.uk}}}
\date{\today}
\maketitle

\begin{abstract} 
We prove a global decomposition result for $\log$-correlated Gaussian fields on the $d$-dimensional torus and use this to derive new small deviations bounds for a class of Gaussian multiplicative chaos measures obtained from Gaussian fields with zero spatial mean on the $d$-dimensional torus. The upper bound is obtained by a modification of the method that was used in \cite{LRV}, and the lower bound is obtained by applying the Donsker--Varadhan variational formula.

We also give the probabilistic path integral formulation of the massless Sinh--Gordon model on a torus of side length $R$, and study its partition function as $R$ tends to infinity. We apply the small deviation bounds for Gaussian multiplicative chaos to obtain lower and upper bounds for the logarithm of the partition function, leading to the existence of a non-zero and finite subsequential infinite volume limit for the free energy.
\end{abstract}

\maketitle

\setcounter{tocdepth}{2}
\tableofcontents

\section{Introduction}

\subsection{Decomposition of $\log$-correlated Gaussian fields}

Let $Z$ be a log-correlated centred Gaussian process on an open domain $D \subset \R^d$ with covariance kernel
\begin{align*}
\E[Z(z)Z(w)] &= \log \frac{1}{|z-w|} + h(z,w)\,,
\end{align*}
where $h$ is a continuous function. Often one would prefer to work with a field with good scaling properties. To this end, there exists several works on decompositions of $\log$-correlated Gaussian fields of the form
\begin{align*}
Z = X + H
\end{align*}
where $X$ is an almost $\star$-scale invariant field and $H$ is Hölder continuous. In \cite{JSW} the authors show that a non-degenerate $\log$-correlated Gaussian field with some regularity assumptions can locally be decomposed in this way. They also prove a global decomposition result without independence of the two fields  $X$ and $H$. In \cite{AJJ} the decomposition in the case when $Z$ and $H$ are independent was generalised to be almost global in the sense that the decomposition holds on any compact subset of $D$. In this article we obtain a global decomposition of this type for $\log$-correlated fields defined on the $d$-dimensional torus. As an application of the decomposition, we prove small deviations estimates for Gaussian multiplicative chaos measures.

\subsection{Small deviations of Gaussian multiplicative chaos measures}

The Gaussian multiplicative chaos (GMC) measure corresponding to the $\log$-correlated Gaussian field $Z$ with parameter $\gamma \in \R$ is formally defined as the exponential of $X$
\begin{align*}
M_{Z,\gamma}(dz) &= e^{\gamma Z(z) - \frac{\gamma^2}{2} \E[Z(z)^2] } \, dz\,.
\end{align*}
Such measures were initially studied by Kahane \cite{Kah}, and resurfaced again in \cite{RoVa, ds}. Later the theory was heavily applied in the path integral construction of Liouville conformal field theory \cite{DKRV}. Some basic properties of GMC measures were already studied by Kahane, including universality and the existence of all negative and some positive moments of the total mass of the measure. Later, large deviations of GMC measures was understood in the works \cite{ldp, modick1, modick2}.

The small deviations behaviour of GMC measures is not yet fully understood. The previous works \cite{Nik, TaWi} concern a case where the GMC measure is built out of a Gaussian field with non-zero spatial average. From the point of view of small deviations this setting is simpler, as fluctuations of the spatial average of the underlying Gaussian field dominate the small deviations behaviour. Some results exist also in the case of fields with zero spatial average, including \cite{LRV}, although the upper bound derived there is not quite optimal. In the case of the one-dimensional GMC measure constructed from the trace of the two-dimensional Gaussian free field, small deviations results follow from the results in \cite{remy1, remy2}, where an explicit formula for the probability distribution of the GMC mass was obtained via CFT methods. Negative moments of GMC measures built out of the Gaussian free field on fractal sets were considered in \cite{GHSS}.

In this article we prove new upper and lower bounds for a class of GMC measures obtained from Gaussian fields with zero spatial average. We apply these estimates to the Sinh--Gordon model to obtain bounds for the free energy.

\subsection{Exponentially interacting quantum field theories and Gaussian multiplicative chaos}

In recent years there has been tremendous progress in applying the theory of GMC measures to path integral constructions of exponentially interacting quantum field theories. Most notable progress has happened in Liouville theory, see the recent surveys \cite{GKR, witten}, but GMC methods have also been applied to (non-affine) Toda field theories \cite{toda}, imaginary Liouville theory \cite{imaginary, chatterjee, imaginary2} and the $\mathbb{H}^3$-Wess--Zumino--Witten model \cite{wzw}.

The exponentially interacting two-dimensional quantum field theories include several interesting models: the Liouville theory is a relatively simple CFT with a continuous spectrum, the compactified imaginary Liouville theory is an example of a logarithmic CFT, the Sine--Gordon, Sinh--Gordon and Bullough--Dodd models are simple examples of integrable quantum field theories. Due to their physical relevance, these models have been studied extensively in the physics and mathematics literature.

In the physics literature, the form factors of the Sinh--Gordon model were studied in \cite{ KM, FMS}, and exact formulas were obtained in \cite{zam, FLZZ, ztba, teschner}. Recently, the nature of the self-duality of the Sinh--Gordon model has been investigated in \cite{KLM, BLC, til}. Rigorous results on the infinite volume limit of the Sinh--Gordon model with a mass term have been obtained in \cite{fro, AHK, BV}. In \cite{GGV} the model was constructed without a mass term on the infinite cylinder by studying the spectral theory of the Hamiltonian of the model. See also \cite{koz1, koz2} for rigorous results in the ($1+1$)-dimensional setting.

In the present article we focus on the Sinh--Gordon model, which is a  two-dimensional exponentially interacting theory obtained by perturbing the Liouville model. The results obtained would also hold for the Bullough--Dodd model, but we do not explicitly include it to keep the notation simpler. The Sinh--Gordon and Bullough--Dodd models are the only integrable perturbations of Liouville theory, see \cite{mussardo, Dor}. We start by constructing the path integral in finite volume, and then study the behaviour of the partition function in the infinite volume limit. By using the small deviations bounds for GMC measures, we obtain upper and lower bounds for (the logarithm of) the partition function, from which we obtain the existence of non-zero and finite subsequential limit for the free energy of the  model in the infinite volume limit.

\subsection{Main results}

Our first result is a decomposition result into a sum of an almost $\star$-scale invariant field and a regular field for non-degenerate $\log$-correlated Gaussian fields with suitable Sobolev regularity.

\begin{theorem}\label{thm:decomposition}
Let $Z:\T^d \to \R$ be a non-degenerate Gaussian field on the $d$-dimensional torus $\T^d$ with the covariance kernel
\begin{align*}
\E[Z(z)Z(w)] &= - \log d_{\T^d}(z,w) + h(z,w)\,,
\end{align*}
with $h \in H^{d+s}(\T^d \times \T^d)$ for some $s>0$ and $d_{\T^d}$ being the distance function on $\T^d$. Then there exists $0\leq t <\infty$ and  Gaussian fields $X^{t}: \T^d \to \R$ and $H: \T^d \to \R$, independent of each other, such that 
\begin{align*}
Z &= X^{t} + H
\end{align*}
and the fields $X^t$ has the covariance kernel (see Section \ref{sec:log_fields} for definitions)
\begin{align*}
\E[X^t(z) X^t(w)] &= \int_t^\infty \rho(e^u d_{\T^d}(z,w))(1-e^{-\xi u}) \, du \,,
\end{align*}
and $H$ is Hölder continuous almost surely.
\end{theorem}

\begin{remark}
In Proposition \ref{lem:coupling-2} we prove an analogous decomposition result, where the almost $\star$-scale invariant field is replaced with a Gaussian field with covariance operator $(- \Delta^{- \frac{d}{2}} + \Delta^{- \frac{d}{2} - \xi}) \PNN$, where $\Delta^{-1}$ is the inverse of the zero-mean Laplace operator on $\T^d$ and $\PNN$ is the projection onto functions with Fourier transform supported outside of $B(0,N) \subset \Z^d$.
\end{remark}
\noindent
Our second result concerns the probability for the total mass of a GMC measure to be very small.

\begin{theorem}\label{sd_ub}

Let $Z$ be as in Theorem \ref{thm:decomposition} and denote
\begin{align*}
\tilde Z(z) &:= Z(z) - \frac{1}{v_{\T^d}(\T^d)} \int_{\T^d} Z(z) \, dv_{\T^d}\,.
\end{align*}
Then for all $\gamma \in (0, \sqrt{2d})$ there exists $\eps_0>0$ such that for all $\eps \in (0,\eps_0)$ the GMC measure $M_{\tilde Z,\gamma}$ corresponding to $\tilde Z$ satisfies
\begin{align*}
 \exp \big( - c_1(d,\gamma) \eps^{- \frac{2d}{\gamma^2}} \big) \leq  \P(M_{\tilde Z,\gamma}(\T^d) < \eps ) & \leq \exp \big( - c_2(d,\gamma) \eps^{- \frac{2d}{\gamma^2}} \big)\,.
\end{align*}
for some $c_1(d,\gamma), c_2(d,\gamma) > 0$.
\end{theorem}
\begin{remark}
Removing the average of the field is crucial for obtaining the upper bound, as otherwise fluctuations of the average of the field would make the probability of the GMC being small much larger. As a result, the lower bound would be trivial without removing the average.
\end{remark}
\noindent
The third result concerns the large volume behaviour of the Sinh--Gordon partition function, formally given on the two-dimensional $R$-torus $\toR$ by the path integral
\begin{align*}
Z_R &= \int \exp \big(- \int_{\toR} ( \tfrac 12 \ab d \varphi \ab_R^2  + 2 \mu \cosh(\gamma \varphi ) ) dv_{R} \big) \, D \varphi\,.
\end{align*}
The rigorous definition of this object will be given in Section \ref{shg_section}.

\begin{theorem}\label{fe_bounds}
For any $\gamma \in (0,2)$, there exists positive constants $f_\gamma$ and $\tilde f_\gamma$ such that for any $R$ large enough we have that
\begin{align*}
  f_\gamma & \leq \frac{- \log Z_R}{ \mu^{\frac{2}{\gamma Q}} R^2} \leq \tilde f_\gamma\,,
\end{align*}
where $Q = \frac{2}{\gamma} + \frac{\gamma}{2}$.
\end{theorem}
\noindent 
\begin{remark}
In the case of the zero-mean hierarchical free field (or branching random walk) Hofstetter and Zeitouni \cite[Theorem 1.5]{HZ} independently obtained a stronger version of our results. 
They are able to show that, in the notation of Theorem \ref{sd_ub},  $\tilde{c}_{\gamma}=c_{\gamma}(1+\smallO(1))$ as $\varepsilon \to 0$. This also translates into a statement that the limit 
$\lim_{R\to \infty} \frac{\log Z_{R}}{R^2}$ exists, see below for a discussion of this problem. Furthermore in \cite{HZ} the authors study correlation function of the zero-mean Liouville model, and obtain decay of correlations, leveraging knowledge about how small deviations of the GMC are achieved.  
Such results would also be of interest the setting of the zero-mean Gaussian free field.  
\end{remark}

\subsection{Outlook} We have demonstrated the possibility to study properties of the infinite volume limit of the massless Sinh--Gordon model via GMC methods. It is then natural to try to expand these methods to study the Sinh--Gordon model further. Given the bounds we derive for the logarithm of the partition function, it is then expected that the limit
\begin{align*}
F_\gamma &:= \lim_{R \to \infty} \frac{- \log Z_R}{R^2}
\end{align*}
exists. This quantity is called the \emph{free energy} in physics literature, and an explicit formula for it has been conjectured, see for example \cite{FLZZ, til} and references therein. As the partition function $Z_R$ will behave essentially like a  Laplace transform
\begin{align*}
Z_R &\sim \E[e^{- R^{2 + \frac{\gamma^2}{2}} \sqrt{M_{\tilde X_1,\gamma}(\tor) M_{\tilde X_1,-\gamma}(\tor)  } }]\,,
\end{align*}
an application of De Bruijn's Tauberian theorem then leads to the conjecture that the limit
\begin{align*}
C_\gamma &:= \lim_{\eps \to 0} \eps^{\frac{4}{\gamma^2}} \log \Pb \big( \sqrt{M_{\tilde X_1,\gamma}(\tor) M_{\tilde X_1,-\gamma}(\tor)  } \less \eps \big)
\end{align*}
exists and is non-zero. The Tauberian theorem also gives an explicit relation between $F_\gamma$ and $C_\gamma$, and thus the conjectural formula for $F_\gamma$ also yields a conjectural explicit formula for the small deviations constant $C_\gamma$. In this way the integrability of the Sinh--Gordon model are related to integrability properties of Gaussian multiplicative chaos.

Beyond the partition function, the integrability of the Sinh--Gordon model is also supposed to lead to an explicit formula for the expected value of the vertex operator 
\begin{align*}
\langle e^{\alpha \varphi(0)} \rangle &:=  \lim_{R \to \infty} \frac{\int e^{\alpha \varphi(0)} \exp \big(- \int_{\toR} ( \tfrac 12 \ab d \varphi \ab_R^2  + 2 \mu \cosh(\gamma \varphi ) ) dv_{R} \big) \, D \varphi}{Z_R}\,.
\end{align*}
Thus, showing the existence of the above limit is another natural open question. Beyond this, major open problems include the existence of a mass-gap in the infinite volume limit and understanding the integrability of the model in a rigorous way. We also mention that it would be interesting to study the model at the critical value $\gamma=2$. In \cite{BLC,til} the self-dual nature of the Sinh--Gordon model has been discussed, and the authors are not fully certain if the model behaves the same in the $\gamma \in(0,\sqrt{2})$ and $\gamma \in [\sqrt{2},2)$ ranges. From the point of view of GMC theory, and our paper, there does not seem to be a big difference between these two ranges. Our analysis does not include the critical case $\gamma=2$. It is speculated in \cite{BLC} that the infinite volume limit might be massless at the critical value.

\emph{Acknowledgements.} We thank Antti Kupiainen, Rémi Rhodes, Janne Junnila, Colin Guillarmou, Trishen Gunaratnam and Roland Bauerschmidt for useful discussions, and Michael Hofstetter and Ofer Zeitouni for discussing their preprint \cite{HZ} with us. JO acknowledges the support from ERC Starting Grant 101042460 “Interplay of structures in conformal and universal random geometry", Austrian Science Fund (FWF) 10.55776/P33083 and from Academy of Finland. NB acknowledges support from ERC Advanced Grant 74148 “Quantum Fields and Probability". MDW acknowledges support from the Royal Society Research Grant RG\textbackslash R1\textbackslash 251187 “Spectral aspects of multiplicative chaos".

\section{Decomposition of logarithmically correlated Gaussian fields}\label{sec:prel}

\subsection{Preliminaries}

\subsubsection{Notations}

Let $\T^d := \R^d/\Z^d$ be the $d$-dimensional unit torus. It inherits the flat Riemannian metric from $\R^d$ and we denote the corresponding volume form by $v_{\T^d}$. The distance function $d_{\T^d}: \T^d \times \T^d \to \R_+$, is given by $d_{\T^d}(z,w) = |(z-w) \mod 1|$, with the convention that for $z \in \T^d$
\begin{align}
z \mod 1 := (z_1 \mod 1, \hdots, z_d \mod 1) \in (- \tfrac{1}{2},\tfrac{1}{2}]^d\,.
\end{align}
The $L^2$-Sobolev spaces are defined for $s > 0$ as
\begin{align*}
H^s(\T^d) &:= \Big\{ f \in L^2(\T^d) : \|f\|_{H^s(\T^d)}^2 := \sum_{k \in \Z^d} (1+|k|^2)^s | \F{f}(k)|^2 < \infty \Big\}\,,
\end{align*}
where $\F{f}$ denotes the Fourier transform of $f$. For $s > 0$ we define the $\| \cdot \|_{H^{-s}(\T^d)}$-norm of $f \in L^2(\T^d)$ by
\begin{align*}
\|f\|_{H^{-s}(\T^d)}^2 &:= \sum_{k \in \Z^d} \frac{|\hat f (k)|^2}{(1+|k|^2)^s}\,,
\end{align*}
which is always finite for $L^2(\T^d)$-functions.

Let $\PNN: L^2(\T^d) \to L^2(\T^d)$ denote the projections on functions with Fourier support outside a ball of radius $N$, that is, 
\begin{align*}
\PNN f (x) &:= \underset{|n| > N}{\sum_{n \in \mathbb{Z}^d}} \hat f (n) e^{2 \pi i n \cdot x}\,.
\end{align*}
We will denote $\PN := I-\PNN$, where $I$ is the identity operator on $L^2(\T^d)$.

The space of continuous functions on $\T^d$ is denoted by $C(\T^d)$, the Fourier transform of a function $f \in L^2(\R^d)$ by $\mathcal{F}_{\R^d}[f]$ and the measure of a set $D \subset \T^d$ by $|D| := \int_D dv_{\T^d}$.

\begin{remark}[Translation invariant kernels]\label{rem:translation_invairance}
Let $k \in L^2(\T^d)$ and define $K: L^2(\T^d) \to L^2(\T^d)$ by 
\begin{align*}
Kf(x) &= \int_{\T^d} k(x-y) f(y) \, dv_{\T^d}(y)\,.
\end{align*}
Then we have that
\begin{align*}
(\widehat{Kf})(n) &= \F{k}(n) \F{f}(n) \,, \quad n \in \Z^d\,,
\end{align*}
that is, $K$ acts diagonally in Fourier space. This implies that $K$ commutes with $\PN$ and $\PNN$.
\end{remark}

\subsubsection{Logarithmically correlated Gaussian fields}\label{sec:log_fields}

Let $C: \T^d \times \T^d \to \R$ be a symmetric positive definite function satisfying
\begin{align}\label{eq:log-kernel}
C(z,w) &= - \log d_{\T^d}(z,w) + h(z,w)\,, \quad (z,w) \in \T^d \times \T^d\,,
\end{align}
where $h \in C(\T^d \times \T^d)$. Then it holds that
\begin{enumerate}
\item There exists a Gaussian field $X$ on $\T^d$ with the covariance kernel $C$.
\item $X \in H^{-s}(\T^d)$ almost surely for any $s > 0$.
\item The integral operator corresponding to the kernel $C$ is a Hilbert--Schmidt operator on $L^2(\T^d)$.
\end{enumerate}
For proofs of these facts, see Section 2 in \cite{JSW2}. The Gaussian field $X$ is called a \emph{logarithmically correlated Gaussian field}, or \emph{log-correlated field} for short, as Gaussianity will always be assumed. We will often use the abuse of notation
\begin{align*}
\E[X(z) X(w)] &:= C(z,w)\,,
\end{align*}
even though the field $X$ is not defined pointwise. 

Important examples of logarithmically correlated Gaussian fields include different variants of the Gaussian free field in two dimensions, and the so-called $\star$-scale invariant fields. The latter are usually defined on $\R^d$ by the covariance kernel
\begin{align}
C_\infty(z,w) &= \int_0^\infty \rho(e^u ( z-w ) )  \, du\,, \quad (z,w) \in \R^d \times \R^d\,,
\end{align}
where $\rho: \R^d \to [0, \infty)$ is a positive definite function satisfying the following properties
\begin{enumerate}
\item $\rho(0) = 1$,
\item $\rho$ is rotationally symmetric $\rho(x) = \rho( (|x|,0,\hdots,0) )$,
\item $\rho$ is supported in the ball $B(0,\frac{1}{2}) \subset \R^d$, 
\item $\rho \in H^{d+s}(\T^d)$ for some $s > 0$.
\end{enumerate}
We call $\rho$ the seed covariance function. We also introduce the almost $\star$-scale invariant Gaussian fields with parameter $\xi \in (0,\infty)$, defined by the covariance kernel
\begin{align}\label{C_xi}
C_\xi(z,w) &:= \int_0^\infty \rho(e^u ( z-w ) ) (1 - e^{-\xi u} ) \, du\,, \quad (z,w) \in \R^d \times \R^d \,.
\end{align}
For $t > 0$ and $\xi \in (0,\infty]$ wee will use the following notations
\begin{align}\label{eq:C_xi_t}
C_\xi^t(z,w) &:= \int_t^\infty \rho(e^u(z-w))(1-e^{-\xi u}) \, du \,, \\
C_{\xi,t}(z,w) &:= \int_0^t \rho(e^u(z-w))(1-e^{-\xi u}) \, du \,.
\end{align}
For $\xi \in (0,\infty)$ we also define $R_\xi = C_\infty - C_\xi$ with $R_{\xi,t}$ and $R_\xi^t$ defined as above. More details on $\star$-scale invariant fields can be found in \cite{JSW,AJJ}.

\begin{remark}[Almost $\star$-scale invariant fields on the torus]\label{rem:star_torus}
Let $\rho:\R^d \to [0,\infty)$ be a seed covariance function. We can define the corresponding almost $\star$-scale invariant Gaussian field on the torus $\mathbb{T}^d = \R^d/\Z^d$ by defining the covariance $C_\xi: \mathbb{T}^d \times \mathbb{T}^d \to \R$ as
\begin{align*}
C_\xi(x,y) &= \int_0^\infty \rho \big( e^u (d_{\T^d}(x,y), 0, \hdots,0) \big) (1 - e^{-\xi u}) \, du \,,
\end{align*}
We define the kernels $C_\xi^t$ and $C_{\xi,t}$ similarly to \eqref{eq:C_xi_t}. We will often slightly abuse the notation by writing $\rho(e^u(x-y))$ instead of $\rho \big( e^u (d_{\T^d}(x,y), 0, \hdots,0)$ for $x,y \in \T^d$.
\end{remark}

\begin{definition}
We say the field $X$ defined in \eqref{eq:log-kernel} is \emph{non-degenerate} if its covariance kernel is a positive definite function, meaning that for all $f \in L^2(\T^d)$ we have that
\begin{align}\label{eq:non_degenerate}
\int_{\T^d \times \T^d} C(z,w) f(z) f(w) \, dv_{\T^d}(z) \, dv_{\T^d}(w)  > 0\,.
\end{align}
\end{definition}
\noindent As we assumed that $\rho$ is a positive definite function, it follows that $C_\infty$, $C_\xi$ and $R_\xi$ for all $\xi \in (0,\infty)$ are covariance kernels of a non-degenerate Gaussian fields, as well as the $t$-cutoff versions of these kernels.

\subsubsection{Gaussian multiplicative chaos} For a logarithmically correlated Gaussian field $X: \T^d \to \R$, we define its Gaussian multiplicative chaos (GMC) measure on $\T^d$ as 
\begin{align}\label{gmc_formal}
M_{X,\gamma}(dz) &= \lim_{\eps \to 0} e^{\gamma X_\eps(z) - \frac{\gamma^2}{2} \E[X_\eps(z)^2]} \, dv_{\T^d}(z)\,,
\end{align}
where $X_\eps$ is a mollification of $X$ in scale $\eps$ and the limit exists weakly in probability for all $\gamma \in (-\sqrt{2d}, \sqrt{2d})$. In the sequel we will sometimes denote GMC measures formally by
\begin{align*}
M_{X,\gamma}(dz) = e^{\gamma X(z) - \frac{\gamma^2}{2} \E[X(z)^2]} \, dv_{\T^d}(z)
\end{align*}
without explicitly referring to mollifications and limits. For more details on GMC measures, see for example \cite{BePo, review}.

\subsection{Proof of Theorem \ref{thm:decomposition}}

In this section we prove the decomposition result formulated in Theorem \ref{thm:decomposition}.

\begin{theorem} \label{lem:coupling}
	Let $Z: \T^d \to \R$ be a non-degenerate logarithmically correlated Gaussian field with the covariance kernel
\begin{align*}
\E[Z(z)Z(w)] &= -\log d_{\T^d}(z,w) + h(z,w)\,,
\end{align*}	
	with  $h \in H^{d + s}(\T^d \times \T^d)$ for some $s>0$. Let $\xi > 0$. Then there exists $t>0$  and a Gaussian field $H$, such that $H$ is Hölder continuous almost surely, and
   \begin{align*}
   Z \overset{d}{=} X^{t}+H\,,
   \end{align*}
   where $X^{t}$ has covariance 
   \begin{align*}
     C_{\xi}^t(z,w) =  \int_{t}^{\infty} \rho \big( e^{u}(z-w) \big) (1-e^{-\xi u}) \, du\,,
   \end{align*}
   and the fields $X^t$ and $H$ are independent of each other.
\end{theorem}

\begin{proof}
The covariance kernel of the $\star$-scale invariant field with seed covariance $\rho$ can be written as
\begin{align*}
C_\infty(z,w) &= -\log d_{\T^d}(z,w) + \int_{d_{\T^d}(z,w)}^1 \big(\rho(t)-1\big) \tfrac{dt}{t} = - \log d_{\T^d}(z,w) + h_\rho(z,w)
\end{align*}
where $h_\rho \in H^{d + s_1}(\T^d \times \T^d)$ for some $s_1 >0$ (see Proposition 4.1. (vi) in \cite{JSW}). This implies that $C_\infty$ satisfies the assumptions made in \eqref{eq:log-kernel}. We define $G:L^2(\T^d) \to L^2(\T^d)$ by
\begin{align*}
Gf(x) = \int_{\T^d} (h(z,w)-h_\rho(z,w)) f(w) \, dv_{\T^d}(z)\,.
\end{align*}
By the assumption we have $h \in H^{d+s_2}(\T^d \times \T^d)$ for some $s_2 > 0$, so the integral kernel of $G$ belongs to $H^{d+s}(\T^d \times \T^d)$ where $s := \min\{s_1,s_2\} > 0$. We denote the covariance kernel of $Z$ by $C$. Thus, as integral operators, we have that
\begin{align*}
C &= C_\infty + G\,.
\end{align*}
Let
\begin{align*}
   G^{<N}:=\PN G \PN \qquad  G^{>N}:= \PNN G \PN + \PN G \PNN + \PNN G \PNN\,.
\end{align*}
Now the covariance operator $C$ decomposes as
   \begin{align}
   C =  C_\infty+G &= C_\infty^t +  C_{\infty,t}+G \nonumber \\
   & =C_\infty^t + C_{\infty,t} \PN + C_{\infty,t} \PNN+G^{<N}+G^{>N} \nonumber \\
   &= C_\infty^t + C_\infty \PN - C_\infty^t \PN + C_{\infty,t} \PNN +G^{<N} + G^{>N} \,.
\end{align}
As $C_\infty$ is translation invariant, it commutes with $\PN$ and $\PNN$, so we can rewrite
\begin{align*}
C_\infty \PN + G^{<N} &= \PN (C_\infty+G) \PN   - \eps R_\xi \PN + \eps R_\xi \PN \\
&= \PN C \PN - \eps R_\xi \PN + \eps R_\xi \PN\,,
\end{align*}
where $\eps > 0$, $\xi > 0$ and $R_\xi$ is the integral operator with the kernel
\begin{align*}
R_\xi(z,w) &= \int_0^\infty \rho(e^u(z-w)) e^{-\xi u} \, du\,.
\end{align*}
We decompose $C^t_\infty$ as 
\begin{align*}
C_\infty^t &= C_\xi^t  + R^t_\xi \,,
\end{align*}
where $R_\xi^t$ is the integral operator with the kernel
\begin{align*}
R_\xi^t(z,w) &= \int_t^\infty \rho(e^u(z-w¨)) e^{-\xi u} \, du \,.
\end{align*}
Thus the covariance is decomposed as
\begin{align}\label{eq:c_decomp}
C &= C_\xi^t   +  \PN ( C  - \eps R_\xi) \PN  + R_\xi^t + G^{>N} +  (   \eps R_\xi - C_\infty^t  ) \PN  +    (C_{\xi,t} + R_{\xi,t})      \PNN\,,
\end{align}
where in the second term we used translation invariance of the integral kernel of $R_\xi$ to commute it with $\PN$. We will show that $\PN ( C  - \eps R_\xi) \PN$ is the covariance operator of a smooth Gaussian field and that $R_\xi^t + G^{>N} +  (   \eps R_\xi - C_\infty^t  ) \PN  +    (C_{\xi,t} + R_{\xi,t}) \PNN$ is the covariance operator of a Gaussian field that almost surely belongs to $H^{\frac{d}{2}+\alpha}(\T^d)$ for some $\alpha > 0$ once $N$, $t$ and $\xi$ are suitably chosen.

By Lemma \ref{lem:auxilary-coupling}, there exists $\delta_1 > 0$ such that for any $f \in L^2(\T^d)$ we have
\begin{align*}
\langle f,  C  f \rangle   > \delta_1 \| f\|_{H^{-\frac{d}{2}}(\T^d)}^2\,.
\end{align*}
By Lemma \ref{lem:R_bounds} (ii) for all $f \in L^2(\T^d)$ we have 
\begin{align*}
\langle f, R_\xi f \rangle & \leq \delta_2 \|f\|_{H^{-d/2}(\T^d)}^2\,.
\end{align*}
Now for all $f \in L^2(\T^d)$ we get 
\begin{align*}
\langle f, \PN (C - \eps R_\xi) \PN f \rangle &\geq (\delta_1 - \eps \delta_2) \|\PN f\|_{H^{-d/2}(\T^d)}^2\,.
\end{align*}
This shows that for $\eps < \frac{\delta_1}{\delta_2}$ the operator $\PN (C - \eps R_\xi) \PN$ is the covariance of a Gaussian field that almost surely belongs to $\PN L^2(\T^d) \subset C^\infty(\T^d)$.

The integral kernel of $R_\xi^t + G^{>N} +  (   \eps R_\xi - C_\infty^t  ) \PN  +  (C_{\xi,t} + R_{\xi,t}) \PNN$ is Hölder continuous, as the integral kernel of $G$ belongs to $H^{d+s}(\T^d \times \T^d)$ for some $s > 0$ and the covariance kernel of $R_\xi$ is Hölder continuous by Proposition 4.1. (ii) in \cite{JSW}. Next we prove the positivity.

First, we use $C_{\xi,t} \geq 0$, $R_\xi^t \geq 0$ and $\eps < 1$ 
\begin{align*}
R_\xi^t + \eps R_\xi \PN + (C_{\xi,t} + R_{\xi,t}) \PNN &= R_\xi^t + \eps R_{\xi,t} \PN + \eps R_{\xi}^t \PN + (C_{\xi,t} + R_{\xi,t}) \PNN \\
 &\geq R_\xi^t + \eps R_{\xi,t} \PN + R_{\xi,t} \PNN \\
& \geq \eps R_\xi^t + \eps R_{\xi,t} \PN + \eps R_{\xi,t} \PNN \\
& = \eps R_\xi\,.
\end{align*}
Thus, it remains to show that the operator $\eps R_\xi  - C_\infty^t \PN + G^{>N}$ is positive. By Lemma \ref{lem:R_bounds} (i), there exists a  $\delta >0 $ such that  $\langle f, R_\xi f \rangle > \delta \| f \|^2_{H^{-d/2-\xi/2}(\T^d)}$ for all $f \in L^2(\T^d)$. Lemmas \ref{lem:C_PN} and \ref{lem:G_bound} give upper bounds for $C_\infty^t \PN$ and $G^{>N}$, respectively, so we get that for all $f \in L^2(\T^d)$
\begin{align*}
\langle f, (\eps R_\xi  - C_\infty^t \PN + G^{>N} ) f \rangle &> \eps \delta \|  f\|_{H^{-d/2-\xi/2}(\T^d)}^2 - c e^{-dt} N^{d+\xi} \|  f \|_{H^{-d/2-\xi/2}(\T^d)}^2  - c  N^{-2 \alpha} \|f\|_{H^{-d/2-s/2+\alpha}(\T^d)}^2\,,
\end{align*}
where $s>0$ and  $\alpha \leq \frac{d+s}{2}$. We take $\alpha = \frac{s-\xi}{2}$, assuming that $\xi \in (0,s)$, where $s>0$ is such that the integral kernel of $G$ belongs to $H^{d+s}(\T^d \times \T^d)$. This leads to 
\begin{align*}
\langle f, (\eps R_\xi  - C_\infty^t \PN + G^{>N} ) f \rangle & \geq  \big( \eps \delta    - c e^{-dt} N^{d + \xi }  - c N^{ -(s-\xi) } \big) \| f\|_{H^{-d/2-\xi/2}(\T^d)}^2 \,,
\end{align*}
for all $f \in L^2(\T^d)$. It follows that by first taking large $N$ and then suitably large $t$, the operator $\eps R_\xi  - C_\infty^t \PN + G^{>N}$ is positive. Thus we have shown that $R_\xi^t + G^{>N} +  (   \eps R_\xi - C_\infty^t  ) \PN  +  (C_{\xi,t} + R_{\xi,t}) \PNN$ is a Hölder continuous and positive kernel, so it is the covariance kernel of a Hölder continuous Gaussian field.
\end{proof}

\begin{proposition}\label{lem:coupling-2}
   Let $Z$ be a Gaussian field on the torus and let $C$ be its covariance kernel. Assume that
	\begin{align*}
	C(z,w) &= - \log d_{\T^d}(z,w) + h(z,w)\,,
\end{align*}	   
	where $h \in H^{d+s}(\T^d \times \T^d)$ for some $s > 0$.
    Then there exists $\xi>0$, $N \in \N$ and Gaussian fields $\bar{X},\bar{H}$ on $\T^d$ such that 
	\begin{enumerate}
	\item $Z = \bar X + \bar H$ and $\bar X$ is independent of $\bar H$
    \item $\bar{H}$ is Hölder continuous almost surely.
    \item The covariance of $\bar{X}$ is given by the operator
   \begin{align*}
     \PNN (-  \Delta^{-\frac{d}{2}} + \Delta^{-\frac{d}{2}-\xi})\,,
   \end{align*}
   where $\Delta^{-1}$ denotes the inverse of the zero-mean Laplace operator on $\T^d$.
	\end{enumerate}
\end{proposition}
\begin{proof}
We write the proof for $d=2$, as the general case follows by replacing $\Delta^{-1}$  by $\Delta^{-\frac{d}{2}}$.

 We can write $C= - \Delta^{-1} + \bar G$ where the integral kernel of $\bar G$ belongs to $H^{2+s}(\T^2 \times \T^2)$. Now we proceed as in the proof of Theorem \ref{lem:coupling}, replacing $C_\infty$ by $- \Delta^{-1}$ and $C_\xi$ by $-\Delta^{-1}+\Delta^{-1-\xi}$ and $G$ by $\bar{G}$.
 
First, we write
 \begin{align}\label{eq:C2}
 C &= - \Delta^{-1} \PNN + \PN (C + \eps \Delta^{-1-\xi}) \PN - \eps \Delta^{-1-\xi} \PN + \bar G^{>N}\,,
 \end{align}
 where $\bar G^{>N} := \bar G - \PN \bar G \PN$ and we used the fact that $\Delta$ commutes with $\PN$. By Lemma \ref{lem:auxilary-coupling} there exists $\delta > 0$ such that 
 \begin{align*}
 \langle f,  C f \rangle & \geq \delta \|f\|_{H^{-1}(\T^2)}^2 
 \end{align*}
 for all $f \in L^2(\tor)$, and we have the bound
\begin{align}\label{eq:Delta_xi_bound}
\langle f,  - \Delta^{-1-\xi} f \rangle &= \sum_{k \in \Z^2 \setminus \{0\}} |\hat f(k)|^2 |k|^{-2-2\xi}  \leq c \sum_{k \in \Z^2 \setminus \{0\}} \frac{|\hat f(k)|^2}{(1+|k|^2)^{1+\xi}} = c\|f\|_{H^{-1-\xi}(\T^2)}^2\,.
\end{align}
Thus,
\begin{align*}
\langle f, \PN (C+ \eps \Delta^{-1-\xi}) \PN f \rangle & \geq (\delta - \eps c) \|\PN f\|_{H^{-1}(\T^2)}^2\,,
\end{align*}
so for small enough $\eps$ the operator $\PN (C+ \eps \Delta^{-1-\xi}) \PN$ is the covariance of a smooth Gaussian field.

For the rest of the terms in \eqref{eq:C2}, we write
\begin{align*}
- \Delta^{-1} \PNN - \eps \Delta^{-1-\xi} \PN + \bar G^{>N} &= (- \Delta^{-1} + \Delta^{-1-\xi}) \PNN  - \Delta^{-1-\xi} \PNN - \eps \Delta^{-1-\xi} \PN + \bar G^{>N} \\
& \geq (- \Delta^{-1} + \Delta^{-1-\xi}) \PNN  - \eps \Delta^{-1-\xi} + \bar G^{>N}\,.
\end{align*}
We have
\begin{align*}
\langle f, - \Delta^{-1-\xi} f \rangle &\geq c \|f\|_{H^{-1-\xi}(\tor)}^2\,.
\end{align*}
By combining this with Lemma \ref{lem:G_bound}, we get
\begin{align*}
\langle f, (- \eps \Delta^{-1-\xi} + \bar G^{>N}) f \rangle & \geq \eps c \|f\|_{H^{-1-\xi}(\tor)}^2 - c N^{-2 \alpha} \|f\|_{H^{-1-s/2+\alpha}(\tor)}^2\,,
\end{align*}
and it suffices to take $\alpha = \frac{s}{2} - \xi$ which is positive for $\xi \in (0,\frac{s}{2})$. It follows that for large enough $N$, $- \eps \Delta^{-1-\xi} + \bar G^{>N}$ is a positive operator with integral kernel belonging to $H^{d+\xi}(\T^d \times \T^d)$ (regularity of the integral kernel of $\Delta^{-1-\xi}$ follows from Theorem 3.3. in \cite{LSSW}), so it is the covariance kernel of a Hölder continuous Gaussian field.

\end{proof}

\subsubsection{Estimates}

\begin{lemma}\label{lem:C_PN}
Let $t>0$, $\xi \in (0,\infty]$, $N \in \N$ and $s > 0$. Then there exists $c > 0$ such that for all $f \in L^2(\T^d)$ we have that
\begin{align*}
\langle f, C_\xi^t P_{<N} f \rangle & \leq c e^{-dt} N^{2s}  \|f\|^2_{H^{-s}(\T^d)}\,.
\end{align*}
\end{lemma}
\begin{proof}
As $C_\xi^t$ commutes with $\PN$, we have
\begin{align*}
\langle f, C_\xi^t P_{<N} f \rangle &=  \langle \PN f, C_\xi^t \PN f \rangle =  \sum_{|k| < N} |\F{f} (k)|^2 \F{C_\xi^t}(k)\,.
\end{align*}
Denote $\rho_u(x) = \rho(e^u x)$. We denote by $\hat \rho_u: \Z^d \to \C$ the Fourier transform of the function $\tilde \rho_u: \T^d \to \R$, $\tilde \rho_u(z) = \rho(e^u(z \mod 1))$. We have for $n\in \mathbb{Z}^d$
\begin{align}\label{eq:rho_fourier}
   \hat \rho_u(n) &=\int_{\T^d} e^{-2 \pi i n \cdot x} \tilde \rho_u(x) \, dv_{\T^d}(x)  =\int_{(-1/2, 1/2)^d} e^{-2 \pi i n \cdot x} \rho(e^u x) \, dx =  \int_{\R^d} e^{-2 \pi i n \cdot e^{-u} x } \rho(x) \, dx \nonumber \\
   &= e^{-du} \mathcal{F}_{\R^d}[\rho](e^{-u}n) \,,
\end{align}
where we used the fact that $\rho$ is supported in $B(0,\frac 12)$. Now we get that 
\begin{align*}
\hat C_\xi^t(n) &= \int_t^\infty e^{-du} \mathcal{F}_{\R^d}[\rho] (e^{-u} n) (1-e^{-\xi u}) \, du  \leq \| \mathcal{F}_{\R^d}[\rho] \|_{L^\infty(\R^d)} \int_t^\infty e^{-du} \, du   \leq c e^{-dt}\,,
\end{align*}
as $\rho \in H^{d+s}(\T^d)$ for some $s >0$ implies that $\| \mathcal{F}_{\R^d}[\rho] \|_{L^\infty(\R^d)} < \infty$. Thus,
\begin{align*}
\left| \sum_{|k| < N} |\F{f} (k)|^2 \F{C_\xi^t}(k) \right| & \leq c e^{-dt} \sum_{k \in \Z^d, |k| < N} |\hat f(k)|^2   \\
& \leq c e^{-dt} (N+1)^{2s} \sum_{|k|<N} \frac{| \F{f}(k)|^2}{(1+|k|^2)^{s}} \\
&\leq c e^{-dt} (N+1)^{2s}  \|f\|^2_{H^{-s}(\T^d)}\,.
\end{align*}
\end{proof}

\begin{lemma}\label{lem:R_bounds}
\begin{enumerate}
\item[(i)] There exists a $\delta_1 > 0$ such that
\begin{align*}
\langle f, R_\xi f \rangle & \geq \delta_1 \|f\|_{H^{-d/2-\xi/2}(\T^d)}^2
\end{align*}
for all $f \in L^2(\T^d)$.

\item[(ii)] There exists a $\delta_2 > 0$ such that
\begin{align*}
\langle f, R_\xi f \rangle & \leq \delta_2 \|f\|_{H^{-d/2}(\T^d)}^2
\end{align*}
for all $f \in L^2(\T^d)$.
\end{enumerate}
\end{lemma}

\begin{proof}
\emph{(i):} See Lemma 4.6. in \cite{JSW}.

\emph{(ii):} By using \eqref{eq:rho_fourier} and $\rho \in H^{d+s}(\T^d)$ for some $s >0$, we can estimate
\begin{align*}
\langle f, R_\xi f \rangle &=\sum_{k \in \Z^d} |\hat f(k)|^2 \hat R_\xi(k) \\
& \leq \sum_{k \in \Z^d} |\hat f(k)|^2 \int_0^\infty e^{-(\xi + d) u} \mathcal{F}_{\R^d}[\rho](e^{-u} k) \, du \\
& \leq \sum_{k \in \Z^d} \frac{|\hat f(k)|^2}{(1+|k|^2)^{\frac{d}{2}}}  \int_0^\infty e^{-(\xi+d) u} c \frac{(1+|k|^2)^{\frac{d}{2}}}{(1+ e^{-2u} |k|^2)^{ \frac{d}{2}}} \, du \\
& \leq c\sum_{k \in \Z^d} \frac{|\hat f(k)|^2}{(1+|k|^2)^{\frac{d}{2}}} \int_0^\infty e^{-\xi u} \, du  \\
& \leq c\xi^{-1} \|f\|_{H^{-d/2}(\T^d)}^2 \,.
\end{align*}
\end{proof}

\begin{lemma}\label{lem:G_bound}
Let $s = \min\{s_1,s_2\}$ where $s_1$ and $s_2$ are such that $\rho \in H^{d+s_1}(\T^d \times \T^d)$ and $h \in H^{d+s_2}(\T^d \times \T^d)$. Then for any $\alpha < \frac{d+s}{2}$  and $f \in L^2(\T^d)$ we have
\begin{align*}
\langle f, G^{>N} f \rangle \leq c  N^{-2\alpha}  \|f\|^2_{H^{-(d+s)/2+\alpha}(\T^d)}\,.
\end{align*}
\end{lemma}

\begin{proof} Let $g = h-h_\rho$ be the integral kernel of $G$ (defined in the beginning of the proof of Lemma \ref{lem:coupling}). By our assumptions we have that $g \in H^{d+s}(\T^d)$. Now, for any $f_1, f_2 \in L^2(\T^d)$ we have that 
\begin{align*}
|\langle f_1, G f_2 \rangle| &= \left| \sum_{n,m \in \Z^d} \F{g}(n,m) \overline{\F{f_1}(n)} \F{f_2}(m)  \right| \\
&\leq \sum_{n,m \in \Z^d} \frac{| \F{f_1}(n)| |\F{f_2}(m)|}{(1+|n|^2)^{(d+s)/4} (1+|m|^2)^{(d+s)/4}}  (1+|n|^2)^{(d+s)/4}  (1+|m|^2)^{(d+s)/4} |\F{g}(n,m)| \\
& \leq  \|f_1\|_{H^{-(d+s)/2}(\T^d)} \|f_2\|_{H^{-(d+s)/2}(\T^d)} \|g\|_{H^{d+s}(\T^d \times \T^d)}\,.
\end{align*}
This implies that for any $f \in L^2(\T^d)$ we have that 
\begin{align*}
|\langle f, G^{>N} f \rangle| &= \big| \langle \PN f, G \PNN f \rangle  + \langle \PNN f, G \PN f \rangle + \langle \PNN f , G \PNN f \rangle \big| \\
& \leq C \|g\|_{H^{d+s}(\T^d \times \T^d)} \|f\|_{H^{-(d+s)/2}(\T^d)} \|\PNN f\|_{H^{-(d+s)/2}(\T^d)} \,.
\end{align*}
Now the claimed result follows by the estimate
\begin{align*}
\|\PNN f \|_{H^{-s}(\T^d)}^2 &= \sum_{k \in \Z^d, |k|>N} \frac{|\F{f}(k)|^2}{(1+|k|^2)^{(d+s)/2}} \leq N^{-2\alpha} \|f\|_{H^{-(d+s)/2+\alpha}(\T^d)}^2
\end{align*}
where $\alpha < \frac{d+s}{2}$.
\end{proof}

\begin{lemma}\label{lem:auxilary-coupling} Let $Z$ be a logarithmically correlated Gaussian field as in Theorem \ref{lem:coupling} and denote its covariance operator by $C$. Then there exists a $\delta>0$ such that 
   \begin{align*}
   \langle f, C f \rangle_{L^2(\T^d)}  \geq \delta \|f\|_{H^{-\frac{d}{2}}(\T^d)}^2\,.
   \end{align*}
\end{lemma}
\begin{proof} 
	The operator $(I- \Delta)^{-\frac{d}{2}}$ on the torus $\T^d$ has an integral kernel of the form
	\begin{align*}
	-\log d_\tor(z,w) + m(z,w)\,,
	\end{align*}
	where $m \in H^{d+2}(\T^d \times \T^d)$. Now we may write $C=(I-\Delta)^{-\frac{d}{2}} +\bar{G}$ where $\bar G$ has an integral kernel $\bar g \in H^{d+ s \land 2}(\T^d \times \T^d)$ for some $s>0$. 
	Then by the assumption that $C$ is non-degenerate, we have
   \begin{align*}
    \langle f ,C f\rangle > 0
   \end{align*}
   for every $f \in L^2(\T^d)$. Assume that the claim does not hold. Then there exist functions $f_{n}\in L^2(\T^d)$ with $\|f_{n}\|_{H^{-\frac{d}{2}}(\T^d)}=1$ such that $\langle f_{n} ,C f_{n}\rangle_{L^2(\T^d)} \to 0$. Now by Banach--Alaoglu we can choose a subsequence of $(f_n)_{n \in \N}$ that converges to $0$ weakly in $H^{-\frac{d}{2}}(\T^d)$.  Then 
   \begin{align*}
   \lim_{n \to \infty} \langle f_n, C f_n \rangle =  \lim_{n \to \infty} \langle f_{n},((I-\Delta)^{-1}+\bar{G})f_{n}\rangle = \lim_{n \to \infty} \|f_{n}\|^2_{H^{-\frac{d}{2}}(\T^d)}+\lim_{n \to \infty} \langle f_{n}, \bar{G} f_{n} \rangle = 1 + \lim_{n \to \infty} \langle f_n, \bar G f_n \rangle\,.
   \end{align*}
   As $\bar G$ is a compact operator, we may choose a further subsequence such that $\bar G f_n \to 0$. This leads to the contradiction that $\lim_{n \to \infty} \langle f_n, C f_n \rangle = 1$.
\end{proof}

\section{Small deviations of Gaussian multiplicative chaos}

\noindent Let $Z: \T^d \to \R$ a logarithmically correlated Gaussian field. Let $V \subset \T^d$ be a Borel set with positive Lebesgue measure and denote
\begin{align}\label{eq_Z_tilde}
\tilde Z(z) &:= Z(z) - \frac{1}{|V|} \int_{V} Z(z) \, dv_{\T^d}(z)\,.
\end{align}
In this section we use a slightly different convention than in \eqref{gmc_formal}, and we denote
\begin{align}\label{eq_M_tilde}
\tilde M_{Z,\gamma} (dz) &= e^{\gamma \tilde Z(z) - \frac{\gamma^2}{2} \E[ Z(z)^2]} \, dv_{\T^d}(z) \,,
\end{align}
that is, the normalization is with respect to $Z$ instead of $\tilde Z$. This does not make a big difference, but will be convenient. In this section we will derive upper and lower bounds for the probability $\P ( \tilde M_{Z,\gamma}(V) < \eps )$.

\subsection{Upper bound}

Our starting point is the following lemma concerning functions with zero mean over some fixed set $D \subset \T^d$.

\begin{lemma}\label{ck}
Let $\alpha \grea 0$ and $D \subset \T^d$ be a Borel set with positive Lebesgue measure. Assume that $Z \in L^1(\T^d)$ satisfies 
\begin{align*}
\int_D Z \, dv_{\T^d} &= 0\,.
\end{align*}
Fix $\kappa \in \N$ large enough so that $ 8 \frac{ (2/e)^\kappa}{1-2/e} \less 1$.
Then one of the following holds
\begin{enumerate}
\item There exists $\beta \grea 0$ (depending only on $\kappa$ and not on $\alpha$) and a Borel set $B \subset D$ with $\ab B \ab \geq \beta \ab D \ab$ such that $Z \geq -\alpha$ on $B$.
\item There exists $n \geq \kappa$ and a Borel set $A_n \subset D$ with $\ab A_n \ab \geq e^{-n} \ab D \ab$ such that $Z \geq 4 \alpha 2^n $ on $A_n$.
\end{enumerate}
\end{lemma}

\begin{proof}
Assume that (1) does not hold. Then there must exist a set $B^c$ such that $\ab B^c \ab \geq (1-\beta) \ab D\ab$ and $Z \less -\alpha$ on $B^c$. Let $B$ be the complement of $B^c$. Then, by the zero-mean property,
\begin{align*}
 \frac{1}{\ab D \ab} \int_B Z \, dv_{\T^d} &= -\frac{1}{\ab D \ab} \int_{B^c} Z \, dv_{\T^d} \geq (1-\beta) \alpha \,.
\end{align*}
Now, if also (2) does not hold, then for all $n \geq \kappa$ we have $\ab \{x \in D: Z \geq 4\alpha 2^n \} \ab \less e^{-n} \ab D \ab$ and thus 
\begin{align*}
\frac{1}{\ab D \ab} \int_B Z \, dv_{\T^d} & \leq \frac{4 \alpha 2^\kappa}{\ab D \ab}  \ab \{ 0 \leq Z \leq 4 \alpha 2^{\kappa}\} \ab + \frac{1}{\ab D \ab} \sum_{n= \kappa}^\infty 4 \alpha 2^{n+1} \ab \{ 4 \alpha  2^n \leq Z \leq 4 \alpha 2^{n+1} \} \ab  \\
& \leq 4 \alpha 2^\kappa \beta  + 8\alpha \sum_{n=\kappa}^\infty 2^n e^{-n} \\
&= \alpha \left( 4 \beta 2^\kappa + 8 ( \tfrac 2e )^\kappa \tfrac{1}{1- \frac{2}{e}}  \right)\,.
\end{align*}
Thus 
\begin{align*}
(1-\beta) \alpha & \leq \alpha \left( 4 \beta 2^\kappa + 8  ( \tfrac 2e )^\kappa \tfrac{1}{1- \frac{2}{e}}  \right)\,.
\end{align*}
As $\beta \to 0$, the left-hand side goes towards $\alpha$ and the right-hand side towards $ 8 \frac{ (2/e)^\kappa}{1-2/e} \alpha$. Thus for $\kappa$ suitably large such that $ 8 \frac{ (2/e)^\kappa}{1-2/e} \less 1$, we get a contradiction.
\end{proof}

\begin{remark}
The small deviations estimate in \cite{LRV} uses a similar lemma. The main difference is that in our setting, the parameters $\beta$ and $\kappa$ will not depend on the parameter $\eps$, which simplifies some parts of the argument. In \cite{LRV} this lemma is applied twice, after which a crude Gaussian estimate is used to terminate the calculation. Because our choice of $\beta$ and $\kappa$ will be independent of $\eps$, we have to continue this iteration many more times (depending on how small $\eps$ is), after which we also terminate by a crude Gaussian estimate. This is one of the differences between the two approaches.
\end{remark}
\noindent

Next we give a proof of Theorem \ref{sd_ub} assuming certain technical estimates, which we will prove afterwards.

\begin{theorem}
Let $Z: \T^d \to \R$ a logarithmically correlated Gaussian field satisfying the assumptions of Theorem \ref{lem:coupling} and $V \subset \T^d$ a Borel set with positive measure. Let $\tilde Z =  Z - \frac{1}{|V|} \int_V Z \, dv_{\T^d}$ and
\begin{align*}
\tilde M_{Z,\gamma}(V) &= \int_V e^{\tilde Z(z) - \frac{\gamma^2}{2} \E[Z(z)^2]} \, dv_{\T^d}(z)\,.
\end{align*}
Then there exists a constant $c_\gamma > 0$ such that 
\begin{align*}
\P \big( \tilde M_{Z,\gamma}(V) < \eps  \big) & \leq \exp \big( -c_\gamma \ab V \ab   \eps^{- \frac{2d}{\gamma^2}} \big)\,.
\end{align*}
In the case $V = \T^d$ we also get that 
\begin{align*}
\P \big( M_{\tilde Z,\gamma}(\T^d) < \eps \big) \leq \exp \big( - \tilde c_\gamma e^{- \frac{2d}{\gamma^2}} \big) 
\end{align*}
for some $\tilde c_\gamma > 0$. 
\end{theorem}

\begin{proof}
We give the proof in the two-dimensional case $d=2$. The generalization to an arbitrary dimension is straightforward.

By Theorem \eqref{lem:coupling}, there exists $\mathbf{t} > 0$ and $\xi > 0$ such that we can decompose $Z$ as a sum of two independent Gaussian fields
\begin{align}\label{eq:Z_decomp}
Z &= X^{\mathbf{t}} + H\,,
\end{align}
where $X^{\mathbf{t}}: \T^d \to \R$ is a $\log$-correlated Gaussian field with the covariance kernel
\begin{align*}
\E[X^{\mathbf{t}}(z) X^{\mathbf{t}}(w)] &= C^\mathbf{t}_\xi(z,w) =  \int_\mathbf{t}^\infty \rho \big( e^u (z-w) \big) (1-e^{-\xi u}) \, du
\end{align*}
and $H: \T^d \to \R$ is Hölder continuous almost surely. Note that $X^{\mathbf{t}}$ has the property that $X^{\mathbf{t}}(z)$ is independent of $X^{\mathbf{t}}(w)$ if $d_{\T^d}(z,w) \grea e^{-\mathbf{t}}$.

For $t > \mathbf{t}$ we will define the field $X_t$  with covariance
\begin{align*}
\E[X_t(z) X_t(w)] & = \int_\mathbf{t}^t \rho \big( e^u (z-w) \big) (1-e^{-\xi u}) \, du \,.
\end{align*}
We denote 
\begin{align*}
\tilde X^t &:= X^t -  \frac{1}{|V|} \int_V X^t \, dv_{\T^d} \,, \quad \tilde X_t := X_t - \frac{1}{|V|} \int_V X_t \, dv_{\T^d}\,.
\end{align*}
Now we can decompose $\tilde X^{\mathbf{t}} = \tilde X_{t} + \tilde X^{t}$ for some $t > \mathbf{t}$, where $\tilde X_t$ is independent of $\tilde X^t$. The independence implies that the GMC measure splits as 
\begin{align*}
\tilde M_{X^{\mathbf{t}},\gamma}(dz) &= e^{\gamma \tilde X^{\mathbf{t}}(z) - \frac{\gamma^2}{2}\E [X^{\mathbf{t}}(z)^2] } \, dv_{\T^d}(z)  = e^{\gamma \tilde X_t(z) - \frac{\gamma^2}{2} \E[X_t(z)^2]} \tilde M_{X^t,\gamma} (dz)\,.
\end{align*}
From now on we will denote 
\begin{align*}
\tilde M^{(t)}_\gamma(dz) &:= \tilde M_{X^t, \gamma}(dz)\,.
\end{align*}
By using the decomposition of $Z$ \eqref{eq:Z_decomp}, we write the measure $\tilde M_{Z,\gamma}$ as
\begin{align*}
\tilde M_{Z,\gamma}(V) &= \int_V e^{\gamma( \tilde X_{t}(z)+ \tilde H(z)) - \frac{\gamma^2}{2} (\E[ X_{t}(z)^2] + \E [H(z)^2]) }  \tilde M_\gamma^{(t)}(dz)\,.
\end{align*}
By  $\E[X_t(z)^2] = \int_0^t (1-e^{-\xi u}) \, du \leq t$ we get that
\begin{align}\label{eq:c_H_bound}
\tilde M_{Z,\gamma}(V)  &\geq c_H  e^{- \frac{\gamma^2}{2} t} \int_V e^{\gamma (\tilde X_{t}(z) + \tilde H(z))}  \tilde M^{(t)}_\gamma(dz) \,.
\end{align} 
where $c_H = e^{- \frac{\gamma^2}{2} \sup_{V} \E [H^2]}$. As the covariance kernel of $H$ is Hölder continuous on $\T^d \times \T^d$ (see proof of Theorem \ref{lem:coupling}), we have $c_H > 0$.

Next we introduce some notation to set up things for applying Lemma \ref{ck}. For a tuple of natural numbers $(n_1,\hdots, n_j) \in \N^j$ we denote
\begin{align*}
\bar n_j &:= (n_1,n_2,\hdots,n_j)\,, \\
\ab \bar n_j \ab &:= \sum_{i=1}^j n_i\,.
\end{align*}
We will also use the convention $\bar n_0 := 0$. We also denote
\begin{align*}
t_j &:= t_0 + \ab \bar n_j \ab \,,
\end{align*}
where $t_0 > \mathbf{t}$ is a scale parameter depending on $\eps$ and $\gamma$ to be fixed later. For $j \in \N$ we define the events
\begin{align*}
E_{\bar n_{j},t_0}(A) &:=  \{ \exists \; B_{\bar n_{j}} \subset A: \; |B_{\bar n_{j}}| \geq \beta \ab A \ab \,, \; \bar X_{j} \geq - \alpha 2^{\ab \bar n_{j} \ab} \text{ on } B_{\bar n_{j}}  \}\,, \\
F_{\bar n_{j+1},t_0}(A) &:= \{ \exists  \;  A_{\bar n_{j+1}} \subset A: \; |A_{\bar n_{j+1}}| \geq e^{-n_{j+1}} \ab A \ab \,, \; \bar X_{j} \geq 4 \alpha 2^{|\bar n_{j+1}|} \text{ on } A_{\bar n_{j+1}}\}\,,
\end{align*}
where $\beta > 0$ is arbitrary, $\alpha > 0$ will be fixed later, $\bar X_0 = \tilde X_{t_0}$ and for $j \geq 1$
\begin{align*}
Y_j(A) &:= \tfrac{1}{|A|} \int_{A} (  \tilde X_{t_j} - \tilde X_{t_{j-1}}) \, dv_{\T^d}\,, \\
\bar X_j &:= (\tilde X_{t_j}-\tilde X_{t_{j-1}}) - Y_j(A)\,.
\end{align*}

Denote $\sd = \{ \tilde M_{Z,\gamma}(V) \leq \eps \}$. We apply Lemma \ref{ck} to the function $\tilde X_{t} + \tilde H$ to obtain
\begin{align}\label{eq_EF}
\Pb \big( \sd \big) &\leq  \Pb( \sd \cap E_{0,t}(V) ) + \Pb(\sd \cap F_{\bar n_1, t}(V) )\,.
\end{align}
We choose $t = t_0$ such that $e^{\frac{\gamma^2}{2} t_0} \eps = c_H e^{-\gamma \alpha} \tfrac{ \beta |V|}{2} $, where $\alpha, \beta \grea 0$ are the constants appearing in Lemma \ref{ck}, and $\eps$ small enough so that $t_0 > \mathbf{t}$ still holds. Note that $\beta$ is fixed, but $\alpha$ is arbitrary for now. By \eqref{eq:c_H_bound}, the first term in \eqref{eq_EF} is bounded by
\begin{align*}
\Pb( \sd \cap E_{0,t_0}(V) )  &= \Pb \big( \big\{ c_H e^{- \frac{\gamma^2}{2} t_0} \int_V e^{\gamma (\tilde X_{t_0}(z) + \tilde H(z))}  \tilde M^{(t_0)}_\gamma(dz) \leq \eps \big \} \cap E_{0,t_0}(V)  \big) \\
& \leq \Pb \big( \big\{ c_H e^{- \frac{\gamma^2}{2} t_0} e^{- \gamma \alpha} \tilde M_\gamma^{(t_0)}(B_0) \leq \eps \big \} \cap E_{0,t_0}(V)  \big) \\
&\leq \underset{|B_0| \geq \beta |V|}{\sup_{B_0 \in \mathcal{B}(V)}} \Pb \big( \tilde M_\gamma^{(t_0)}(B_0) \leq \tfrac{\beta \ab V \ab}{2} \big)\,,
\end{align*}
where $\mathcal{B}(V)$ is the set of Borel subsets of $V$. We have
\begin{align*}
\underset{|B_0| \geq \beta |V|}{\sup_{B_0 \in \mathcal{B}(V)}} \Pb \big( \tilde M_\gamma^{(t_0)}(B_0) \leq \tfrac{\beta \ab V \ab}{2} \big) & \leq \underset{|B_0| \geq \beta |V|}{\sup_{B_0 \in \mathcal{B}(V)}} \Pb \big( \tilde M_\gamma^{(t_0)}(B_0) \leq \tfrac{|B_0|}{2} \big)\,.
\end{align*}
As we have $|B_0| \geq \beta |V| \geq e^{-2t_0}$ for $t_0$ large (i.e. $\eps$ small), we can bound this by the concentration inequality from Lemma \ref{concentration} to obtain
\begin{align*}
\underset{|B_0| \geq \beta |V|}{\sup_{B_0 \in \mathcal{B}(V)}} \Pb \big( \tilde M_\gamma^{(t_0)}(B_0) \leq \tfrac{|B_0|}{2} \big) & \leq \underset{|B_0| \geq \beta |V|}{\sup_{B_0 \in \mathcal{B}(V)}} \exp \big( -c |B_0| e^{2t_0} \big) \leq \exp \big( - c \beta |V| e^{2t_0} \big)
\end{align*}
for some $c \grea 0$. To bound the second term in \eqref{eq_EF}, we start by writing
\begin{align*}
\Pb(\sd \cap F_{\bar n_1, t_0}(V) ) &\leq \Pb \big( \big\{  e^{\gamma 4 \alpha 2^{n_1}} \tilde M_\gamma^{(t_0)}(A_1) \leq \tfrac{\beta |V|}{2}  e^{-\gamma \alpha} \big\} \cap F_{\bar n_1, t_0}(V) \big)\,.
\end{align*}
By applying Lemma \ref{XY} we get 
\begin{align*}
&\Pb \big( \big\{  e^{\gamma 4 \alpha 2^{n_1}} \tilde M_\gamma^{(t_0)}(A_1) \leq \tfrac{\beta |V|}{2}  e^{-\gamma \alpha} \big\} \cap F_{\bar n_1, t_0}(V) \big) \\
&\leq \Pb \big( \big\{ e^{\gamma 4 \alpha  2^{n_1} - \frac{\gamma^2}{2} n_1}  e^{\gamma Y_1  } \int_{A_1} e^{\gamma \bar X_1} d \tilde M_\gamma^{(t_1)} \leq \tfrac{\beta |V|}{2}  e^{-\gamma \alpha} \big\} \cap F_{\bar n_1, t_0}(V) \big)\,.
\end{align*}
Now by, Lemma \ref{iteration} there exists $A_2$ and $n_2$ such that 
\begin{align*}
  \sd \cap F_{\bar n_1, t_0}(V)
  \subset& H(V,t_0) \cup  \{ e^{\gamma 4 \alpha 2^{| \bar n_{2} | }  - \frac{\gamma^2}{2} |\bar n_{2}| } e^{\gamma Y_{2}(A_{2})}\int_{A_{2}} e^{\gamma \bar X_{2}} d \tilde M_\gamma^{(t_{2})}\leq \tfrac{\beta \ab V \ab }{2} e^{-\gamma \alpha} \}\,,
  \end{align*}
  where   $H(V,t_0)$ is an event satisfying
\[
    \Pb (H(V,t_0)) \leq C_1 \exp(-C_2 \alpha^2 \ab V \ab  (\tfrac{4}{e})^{| \bar n_1|} e^{2t_{0}}) +  C_1 \exp \big( -C_2 \ab V \ab  e^{2 t_0 + | \bar n_{1}|} \big) \,.
\]
Now, by iterating  the application of Lemma \ref{iteration}, we get that 
\begin{align*}
\Pb(\sd \cap F_{\bar n_1, t_0}(V) ) &\leq C  \sum_{j=1}^{J-1}  \Big( \exp \big( -c \ab V \ab  e^{2t_0 + |\bar n_j|} \big) + \exp \big( -c \ab V \ab  (\tfrac{4}{e})^{\ab \bar n_j \ab }e^{2(t_0 + \ab \bar n_{j-1} \ab )}  \big)  \Big) \\
& \quad + \underset{|A_J| \geq e^{- \ab \bar n_J \ab } \ab V \ab  }{\sup_{A_J \in \mathcal{B}(V)}}  \Pb \big(  e^{\gamma 4 \alpha 2^{|\bar n_J|} - \frac{\gamma^2}{2} | \bar n_J|}  e^{\gamma Y_J  } \int_{A_J} e^{\gamma \bar X_J} d \tilde M_\gamma^{(t_j)} \leq \tfrac{\beta|V|}{2}  e^{-\gamma \alpha} \big) \\
& \leq C \Big( \exp \big( -c \ab V \ab   e^{2t_0 } \big) + \exp \big( -c \ab V \ab  e^{2t_0}  \big)  \Big) \\
& \quad + \underset{|A_J| \geq e^{- \ab \bar n_J \ab } \ab V \ab  }{\sup_{A_J \in \mathcal{B}(V)}}  \Pb \big(  e^{\gamma 4 \alpha 2^{|\bar n_J|} - \frac{\gamma^2}{2} | \bar n_J|}  e^{\gamma Y_J  } \int_{A_J} e^{\gamma \bar X_J} d \tilde M_\gamma^{(t_j)} \leq \tfrac{\beta|V|}{2}  e^{-\gamma \alpha} \big)\,.
\end{align*}
Now, by applying Lemma \ref{ck} we see that 
\begin{align*}
&\{ e^{\gamma 4 \alpha 2^{|\bar n_J|} - \frac{\gamma^2}{2} |\bar n_J|}  e^{\gamma Y_J} \int_{A_J} e^{\gamma \bar X_J} d \tilde M_\gamma^{(t_j)} \leq \tfrac{ \beta |V|}{2}  e^{-\gamma \alpha} \} \\&\subseteq \big(E_{\bar n_{J}}(A_J)\cap  \{ e^{\gamma 4 \alpha 2^{|\bar n_J|} - \frac{\gamma^2}{2} |\bar n_J|}  e^{\gamma Y_J} \int_{A_J} e^{\gamma \bar X_J} d \tilde M_\gamma^{(t_j)} \leq \tfrac{ \beta |V|}{2}  e^{-\gamma \alpha} \}\big) \cup F_{n_{J+1}}(A_{n_{J}})\,.
\end{align*}
It is shown in the proof of Lemma \ref{iteration} below that 
\begin{align*}
 \Pb \big(E_{\bar n_{J}}(A_J)\cap  \{ e^{\gamma 4 \alpha 2^{|\bar n_J|} - \frac{\gamma^2}{2} |\bar n_J|}  e^{\gamma Y_J} \int_{A_J} e^{\gamma \bar X_J} d \tilde M_\gamma^{(t_j)} \leq \tfrac{ \beta |V|}{2}  e^{-\gamma \alpha} \}\big) \leq  \frac{c_3}{\alpha |V|} \exp \big( - c \alpha^2 |V| (\tfrac{4}{e})^{|\bar n_j|} e^{2 t_{j}} \big)\,,
\end{align*}
and by Lemma \ref{terminator} that
\begin{align*}
  F_{n_{J+1}}(A_{n_{J}}) & \leq  \exp \big( -c (1+\delta)^{\ab \bar n_J \ab + \kappa} \big) \,, 
\end{align*}
for some $\delta \grea 0$. As $n_j \geq \kappa$ for each $j$, this can be made arbitrarily small by continuing the iteration until a large enough number of steps $J$.

Finally, we show that the bound for $\tilde M_{Z,\gamma}(\T^d)$ implies the bound for $M_{\tilde Z,\gamma}(\T^d)$. Indeed, we have 
\begin{align*}
M_{\tilde Z,\gamma}(\T^d) &= \int_{\T^d} e^{\gamma \tilde Z - \frac{\gamma^2}{2} \E [\tilde Z^2]} \, dv_{\T^d}(z) = \int_{\T^d} e^{\frac{\gamma^2}{2} \E[Z^2 - \tilde Z^2]} \, d \tilde M_{Z,\gamma}(z)\,.
\end{align*}
Where by $\E[Z^2- \tilde Z^2]$ we denote the function
\begin{align*}
\E[Z(x)^2 - \tilde Z(x)^2] &:= 2 \int_{\T^d} \E[Z(x) Z(z)] \, dv_{\T^d}(z) - \int_{\T^d \times \T^d} \E[Z(z)Z(w)] \, dv_{\T^d}(z) dv_{\T^d}(w)\,.
\end{align*}
By the assumptions on the field $Z$, this is a bounded function on $\T^d$ and the bound for $\P(M_{\tilde Z,\gamma}(\T^2) < \eps)$ follows.
\end{proof}

\noindent
Next we prove the technical lemmas used in the previous proof.

\begin{lemma}\label{iteration}
Let $\bar n_j \in \N^{j}$ and $A_j \subset V$ be a Borel set with $\ab A_j \ab \geq e^{-\ab \bar n_j \ab } \ab V \ab $. Let $\gamma \in (0,2)$ and $\alpha \grea 0$ be large enough so that $\gamma \alpha + \gamma 2 \alpha 2^n - \frac{\gamma^2}{2} n \geq n$ for all $n \in \N$. Then there exists a natural number $n_{j+1} \geq \kappa$, a  Borel set $A_{j+1} \subset A_j$ with $\ab A_{j+1} \ab \geq e^{- \ab \bar n_{j+1} \ab} \ab V \ab$, and an event $H(A_j,t_j)$ such that 
\[
  \Pb (H(A_j,t_j)) \leq C_1 \exp(-C_2 \alpha^2 \ab V \ab  (\tfrac{4}{e})^{| \bar n_j|} e^{2t_{j-1}}) +  C_1 \exp \big( -C_2 \ab V \ab  e^{2 t_0 + | \bar n_{j}|} \big) \,,
  \]
  and
\begin{align*}
  & \{ e^{\gamma 4\alpha 2^{| \bar n_{j} | }  - \frac{\gamma^2}{2} |\bar n_{j}| } e^{\gamma Y_{j}(A_j)}\int_{A_{j}} e^{\gamma \bar X_{j}} d \tilde M_\gamma^{(t_{j})} \leq \tfrac{\beta \ab V \ab }{2} e^{-\gamma \alpha}\}  \\
    &\quad \subseteq H(A_j,t_j) \cup  \{ e^{\gamma 4 \alpha 2^{| \bar n_{j+1} | }  - \frac{\gamma^2}{2} |\bar n_{j+1}| } e^{\gamma Y_{j+1}(A_{j+1})}\int_{A_{j+1}} e^{\gamma \bar X_{j+1}} d \tilde M_\gamma^{(t_{j+1})}\leq \tfrac{\beta \ab V \ab }{2} e^{-\gamma \alpha} \}\,.
  \end{align*}
The constants $C_1$ and $C_2$ do not depend on the set $A_j$.

\end{lemma}

\begin{proof}

We split into the events $G_j$ and $G_j^c$, where
\begin{align*}
G_j &= \{ Y_j(A_j) \leq - \alpha 2^{|\bar n_j|} \}\,.
\end{align*}
By Gaussian estimate, we have that
\begin{align*}
\Pb(G_j) &\leq  c_1 \frac{\E[Y_j(A_j)^2]}{\alpha 2^{|\bar n_j|}} \exp \left( -  \frac{\alpha^2 2^{2|\bar n_j|}}{2 \E[ Y_j(A_j)^2]} \right)\,.
\end{align*}
By Lemma \ref{y_lemma}, we have that
\begin{align*}
\E[ Y_j(A_j)^2] & \leq c_2 \frac{e^{-2 t_{j-1}}}{\ab A_j \ab }\,,
\end{align*}
leading to 
\begin{align*}
\Pb(G_j) &\leq \frac{c_3}{\alpha 2^{|\bar n_j|} e^{2 t_{j-1}} |A_j|} \exp( - \tfrac 12 \alpha^2 2^{2|\bar n_j|} e^{2t_{j-1}} |A_j| ) \\
& \leq \frac{c_3}{\alpha 2^{|\bar n_j|} e^{2 t_{j-1}} e^{-|\bar n_j|}|V|} \exp(- \tfrac 12 \alpha^2 2^{2|\bar n_j|} e^{2 t_{j-1} } e^{-|\bar n_j|} \ab V \ab   ) \\
&  = \frac{c_3}{\alpha |V|  ( \frac{2}{e})^{|\bar n_j|} e^{2 t_{j-1}} } \exp(- \tfrac 12 \alpha^2 \ab V \ab  (\tfrac{4}{e})^{| \bar n_j|} e^{2t_{j-1}}) \\
& \leq \frac{c_3}{\alpha |V|} \exp \big( - c \alpha^2 |V| (\tfrac{4}{e})^{|\bar n_j|} e^{2 t_{j-1}} \big)\,.
\end{align*}
On the other hand, in $G_j^c$ we have 
\begin{align*}
& \Pb \big( G_j^c \cap \{ e^{\gamma 4 \alpha 2^{| \bar n_{j} | }  - \frac{\gamma^2}{2} |\bar n_{j}| } e^{\gamma Y_{j}(A_j)}\int_{A_{j}} e^{\gamma \bar X_{j}} d \tilde M_\gamma^{(t_{j})} \leq \tfrac{\beta \ab V \ab }{2} e^{-\gamma \alpha} \} \big) \\
& \leq \Pb \big( e^{\gamma 3 \alpha 2^{| \bar n_{j} | }  - \frac{\gamma^2}{2} |\bar n_{j}| } \int_{A_{j}} e^{\gamma \bar X_{j}} d \tilde M_\gamma^{(t_{j})} \leq \tfrac{\beta \ab V \ab }{2} e^{-\gamma \alpha} \big)\,,
\end{align*}
and then we apply Lemma \ref{ck} to $\bar X_j: A_j \to \R$. We get the events $E_{\bar n_j, t_0}(A_j)$ and $F_{\bar n_{j+1}, t_0}(A_j)$. The first event can be bounded by
\begin{align*}
& \Pb \big( E_{\bar n_j, t_0}(A_j) \cap \{ e^{\gamma 3 \alpha 2^{| \bar n_{j} | }  - \frac{\gamma^2}{2} |\bar n_{j}| } \int_{A_{j}} e^{\gamma \bar X_{j}} d M^{(t_{j})} \leq \tfrac{\beta \ab V \ab }{2} e^{-\gamma \alpha} \} \big) \\
& \leq \underset{\ab B_{\bar n_j} \ab \geq e^{- \ab \bar n_j \ab } \beta \ab V \ab  }{\sup_{B_{\bar n_j} \in \mathcal{B}(V)}} \Pb \big( e^{\gamma 2 \alpha 2^{| \bar n_{j} | }  - \frac{\gamma^2}{2} |\bar n_{j}| }  \tilde M_\gamma^{(t_{j})}(B_{\bar n_j}) \leq \tfrac{\beta \ab V \ab }{2} e^{-\gamma \alpha} \big)\,.
\end{align*}
Now, by the assumption on $\alpha$, for any $B_{\bar n_j}$ with $\ab B_{\bar n_j} \ab \geq e^{- \ab \bar n_j \ab } \beta \ab V \ab$ we get 
\begin{align*}
\Pb \big( e^{\gamma 2 \alpha 2^{| \bar n_{j} | }  - \frac{\gamma^2}{2} |\bar n_{j}| }  \tilde M_\gamma^{(t_{j})}(B_{\bar n_j}) \leq \tfrac{\beta \ab V \ab }{2} e^{-\gamma \alpha} \big) & \leq \Pb \big( \tilde M_\gamma^{(t_j)}(B_{\bar n_j}) \leq \tfrac{\beta \ab V \ab }{2} e^{-|\bar n_j|} \big) \\
& \leq \Pb \big( \tilde M_\gamma^{(t_j)}(B_{\bar n_j}) \leq \tfrac 12 |B_{\bar n_j}| \big)\,.
\end{align*}
Now, since $|B_{\bar n_j}| \geq e^{-|\bar n_j|} \beta |V|  \geq e^{-2 t_j}$, we may apply the concentration inequality given by Lemma \ref{concentration} to obtain the bound
\begin{align*}
\underset{\ab B_{\bar n_j} \ab \geq e^{- \ab \bar n_j \ab } \beta \ab V \ab  }{\sup_{B_{\bar n_j} \in \mathcal{B}(V)}} \Pb \big( \tilde M_\gamma^{(t_j)}(B_{\bar n_j}) \leq \tfrac 12 |B_{\bar n_j}| \big) & \leq \underset{\ab B_{\bar n_j} \ab \geq e^{- \ab \bar n_j \ab } \beta \ab V \ab  }{\sup_{B_{\bar n_j} \in \mathcal{B}(V)}} \exp \big( -c   |B_{\bar n_j}| e^{2 t_j}  \big) \\
& \leq \exp \big( -c \beta \ab V \ab e^{- \ab \bar n_j \ab } e^{2 (t_0+ \ab \bar n_j \ab )}  \big) \\
&= \exp \big( -c \beta \ab V \ab  e^{2t_0 + \ab \bar n_j \ab }  \big)\,.
\end{align*}
For the other event, by using the definition of $F_{\bar n_{j+1}, t_0}(A_j)$ and Lemma \ref{XY}, we get that there exists a set $A_{j+1}$ with $|A_{j+1}|\geq e^{|\bar n_j|}|V|$ so that the following inclusions hold 
\begin{align*}
& F_{\bar n_{j+1}, t_0}(A_j) \cap \{ e^{\gamma 3 \alpha 2^{| \bar n_{j} | }  - \frac{\gamma^2}{2} |\bar n_{j}| } \int_{A_{j}} e^{\gamma \bar X_{j}} d \tilde M_\gamma^{(t_{j})} \leq \tfrac{\beta \ab V \ab }{2} e^{-\gamma \alpha} \}  \\
& \subset \big\{ e^{\gamma 3 \alpha 2^{| \bar n_{j} | }  - \frac{\gamma^2}{2} |\bar n_{j}| + \gamma 4\alpha 2^{|\bar n_{j+1}|} } \tilde M_\gamma^{(t_{j+1})}(A_{j+1}) \leq \tfrac{\beta \ab V \ab }{2} e^{-\gamma \alpha} \big\} \\
&  \subset \big\{ e^{\gamma 3 \alpha 2^{| \bar n_{j} | }  - \frac{\gamma^2}{2} |\bar n_{j}| + \gamma 4 \alpha 2^{|\bar n_{j+1}|} - \frac{\gamma^2}{2} n_{j+1} } e^{\gamma Y_{j+1}(A_{j+1})} \int_{A_{j+1}}e^{\gamma \bar X_{j+1}}  d \tilde M_\gamma^{(t_{j+1})} \leq \tfrac{\beta \ab V \ab }{2} e^{-\gamma \alpha} \big\} \\
& \subset \big\{ e^{\gamma 4 \alpha 2^{|\bar n_{j+1}\ab } - \frac{\gamma^2}{2} |\bar n_{j+1}|} e^{\gamma Y_{j+1}(A_{j+1})} \int_{A_{j+1}}e^{\gamma \bar X_{j+1}}  d \tilde M_\gamma^{(t_{j+1})} \leq \tfrac{\beta \ab V \ab }{2} e^{-\gamma \alpha} \big\}\,.
\end{align*}
Thus we obtain the desired result.

\end{proof}

\begin{lemma}\label{terminator}
We have
\begin{align*}
\Pb \big(F_{\bar n_{J+1}, t_0}(A_{\bar n_J})  \big) & \leq \exp \big( - c (1+\delta)^{\ab \bar n_J \ab + \kappa }  \big)
\end{align*}
for some $\delta \grea 0$. The constant $c$ is independent of the set $A_{\bar n_j}$.
\end{lemma}

\begin{proof} Recall the definition
\begin{align*}
F_{\bar n_{j+1}, t_0}(A_{\bar n_j}) = \{ \exists \; n_{j+1} \geq \kappa\,,   \;  A_{\bar n_{j+1}} \subset A_{\bar n_j}: \; |A_{\bar n_{j+1}}| \geq e^{-n_{j+1}} \ab A_{\bar n_j} \ab \,, \; \bar X_{j} \geq 4 \alpha 2^{|\bar n_{j+1}|} \text{ on } A_{\bar n_{j+1}}\}\,.
\end{align*}
We use a crude Gaussian estimate. We have
\begin{align*}
\E[ | \{ x  \in A_{\bar n_J} : \bar X_{J} \geq  4 \alpha 2^{| \bar n_{J+1}|} \}| ] & = \E \big[ \int_{A_{\bar n_J}} \mathbf{1}_{\{ \bar X_{J}(x) \geq  4 \alpha 2^{| \bar n_{J+1}|} \}} \, dx \big] \\
&= \int_{A_{\bar n_J}} \Pb \big( \bar X_{J}(x) \geq  4 \alpha 2^{| \bar n_{J+1}|} \big) \, dx   \\
& \leq  \int_{A_{\bar n_J}}  \frac{\E[ \bar X_J(x)^2]}{4 \alpha 2^{\ab \bar n_{J+1}\ab }} \exp \big( - c_2 \alpha^2 2^{2 | \bar n_{J+1}|}/ \E [ \bar X_{J}(x)^2] \big) \, dx  \,.
\end{align*}
By Lemma \ref{bar_X} we have
\begin{align*}
\E [\bar X_{J}(x)^2] &\leq n_J + \tfrac{c}{\ab A_J\ab} e^{-2 t_{J-1}} \leq n_J + \tfrac{c}{\ab V \ab} e^{ \ab \bar n_J \ab-2t_{J-1}}    = n_J + \tfrac{c}{\ab V \ab} e^{-2t_0 - \ab \bar n_{J-1} \ab + n_J}\,.
\end{align*}
Thus we get
\begin{align*}
\E[ | \{ x  \in A_{\bar n_J} : \bar X_{J} \geq  4 \alpha 2^{| \bar n_{J+1}|} \}| ] & \leq |A_{\bar n_J}| \frac{n_J + \frac{c}{|V|} e^{n_J}}{4 \alpha 2^{|\bar n_{J+1}|}} \exp \Big( - c_2 \alpha^2  \frac{2^{2 | \bar n_{J+1}|}}{n_J + \tfrac{c}{\ab V \ab} e^{-2t_0 - \ab \bar n_{J-1} \ab + n_J}}  \Big) \\
& \leq C_1 |A_{\bar n_J}| \exp \big( - C_2 (\tfrac{4}{e})^{|\bar n_{J+1}|} \big)\,.
\end{align*}
Now, by Markov's inequality we get 
\begin{align*}
\Pb \big( F_{\bar n_{J+1}, t_0}(A_{\bar n_J}) \big) & \leq \sum_{n_{J+1} \geq \kappa} \Pb \big( | \{  x \in A_{\bar n_J} : \bar X_j \geq 4 \alpha 2^{|\bar n_{J+1}|} \} | \geq e^{-n_{J+1}} |A_{\bar n_J}| \big) \\
& \leq \sum_{n_{J+1} \geq \kappa} e^{n_{J+1}} |A_{\bar n_J}|^{-1} \E [ | \{  x \in A_{\bar n_J} : \bar X_j \geq 4 \alpha 2^{|\bar n_{J+1}|} \} | ]  \\
& \leq  \sum_{n_{J+1} \geq \kappa} C_1  \exp \big( - C_2 (\tfrac{4}{e})^{|\bar n_{J+1}|} \big) \\
& \leq C_1' \exp \big( -C_2'  (\tfrac{4}{e})^{|\bar n_{J}| + \kappa} \big)
\end{align*}
and the claim follows.
\end{proof}

\subsubsection{Estimates for almost $\star$-scale invariant fields}

In this section we will often denote $dv_{\T^d}(x)$ by $dx$.

\begin{lemma}\label{y_lemma}
For any $A \subset \T^2$ with positive Lebesgue measure and $t,s \grea 0$ we have that 
\begin{align*}
\int_{A} \int_t^{t+s} \rho(e^u (x-y))(1-e^{-\xi u}) \, du \,  dy  &\leq c  e^{-2t} \big( 1 - e^{-2s} \big) \,.
\end{align*}
In particular, for $Y := \frac{1}{|A|} \int_A \big( \tilde X_{t+s}(x) - \tilde X_t(x) \big) \, dx$ we have that 
\begin{align*}
\E[Y^2] & \leq c \frac{e^{-2t}}{|A|}\,.
\end{align*}
\end{lemma}
\begin{proof} 
First, by $1-e^{-\xi u} \leq 1$, it suffices to focus on the case of exactly $\star$-scale invariant fields. By a change of variable in $u$ we have that 
\begin{align*}
 \int_t^{t+s} \rho(e^u(x-y)) \, du  &=   \int_0^{s} \rho(e^{u} e^t(x-y)) \, du \,.
 \end{align*}
 By using the facts that $\rho$ is bounded and supported in the unit ball, we get the upper bound
 \begin{align*}
\int_{A} \int_0^{s} \rho(e^{u} e^t(x-y)) \, du \, dy & \leq c_1 \int_{A} \int_0^{s} \,  \mathbf{1}_{|x-y| < e^{-u-t}} \, du \, dy \\
&= c_1 \int_0^s   \int_{y \in A \cap B(x,e^{-u-t})} dy \, du \\
  & \leq c_1  \int_0^s  |B(x,e^{-u-t}) | \, du  \\
  & \leq c_2  e^{-2t} \int_0^s e^{-2u} \, du \\
  &= \tfrac 12 c_2  e^{-2t} \big( 1 - e^{-2s} \big)\,.
\end{align*}

\noindent By definition
\begin{align*}
\E[ \tilde X_{s_1}(x) \tilde X_{s_2}(y) ] &= \int_0^{s_1 \land s_2} \widetilde Q_u(x,y) \, du \,,
\end{align*}
where $\widetilde Q_u$ is the covariance kernel obtained from
\begin{align*}
Q_u(x,y) &:= \rho(e^u(x-y)) (1-e^{-\xi u}) \,,
\end{align*}
by forcing it to have zero mean over $V$, that is,
\begin{align}\label{tildeq}
\widetilde Q_u(x,y) &:=  Q_u(x,y) - \frac{1}{\ab V \ab} \big( \int_V Q_u(x',y) \, dx' + \int_V Q_u(x,y') \, dy' \big) + \frac{1}{\ab V \ab^2} \int_{V^2} Q_u(x',y') \,dx' \, dy'\,.
\end{align}
Thus,
\begin{align*}
\E[Y^2] &= \frac{1}{\ab A \ab^2} \int_{A^2} \E[ \big( \tilde X_{t+s}(x)- \tilde X_t(x) \big) \big( \tilde X_{t+s}(y)- \tilde X_t(y) \big)] \, dx \, dy \\
&= \frac{1}{\ab A \ab^2} \int_{A^2} \big( \int_0^{t+s} \widetilde Q_u(x,y) \, du  - 2 \int_0^{t}\widetilde Q_u(x,y) du + \int_0^t \widetilde Q_u(x,y) du   \big) \, dx \, dy \\
&= \frac{1}{\ab A \ab^2} \int_{A^2} \int_t^{t+s} \widetilde Q_u(x,y) \,dx \,dy \,  du\,.
\end{align*}
Now, we plug in \eqref{tildeq}, use positivity of $Q_u$ and the previous computation, leading to the upper bound 
\begin{align*}
\E[Y^2] &\leq \frac{1}{\ab A \ab^2} \int_{A^2} \int_t^{t+s} \rho(e^u (x-y) ) (1-e^{-\xi u}) \, du \, d^2x \, d^2y  \\
& \quad + \frac{1}{\ab V \ab^2}\int_{V^2} \int_t^{t+s} \rho(e^u (x-y) )(1-e^{-\xi u}) \, du \, d^2x \, d^2y \\
& \leq  \frac{c}{\ab A \ab} e^{-2t}(1-e^{-2s})\,.
\end{align*}

\end{proof}

\begin{lemma}\label{XY}
Let $t,s > 0$ and $A \subset V$ be a Borel set with positive Lebesgue measure. We have
\begin{align*}
\tilde M^{(t)}_\gamma(A) &\geq  e^{\gamma Y - \frac{\gamma^2}{2}s} \int_A e^{\gamma \bar X} d\tilde M^{(t+s)}_\gamma\,,
\end{align*}
where
\begin{align*}
Y &:= \frac{1}{|A|} \int_A (\tilde X_{t+s} - \tilde X_t) \, d^2x\,, \\
\bar X &:= \tilde X_{t+s} - \tilde X_t - Y\,.
\end{align*}
\end{lemma}

\begin{proof}
Recall that
\begin{align*}
\tilde M^{(t)}_{\gamma}(dz) &:= e^{\gamma \tilde X^{(t)} (z) - \frac{\gamma^2}{2} \E[  X^{(t)}(z)^2]} \, dz\,.
\end{align*}
We use $X^{(t)} \overset{d}{=} (X_{t+s}-X_t) + X^{(t+s)}$. Because $\rho$ is a positive definite function, the two terms on the right-hand side are independent, so we obtain 
\begin{align*}
\tilde M_\gamma^{(t)}(A) &= \int_A e^{\gamma (\tilde X_{t+s} - \tilde X_t) - \frac{\gamma^2}{2} \E(X_{t+s} -  X_t)^2} d \tilde M_\gamma^{(t+s)}\,.
\end{align*}
We have
\begin{align}\label{eq_trivial}
\E[(X_{t+s} -  X_t)^2] &= \int_t^{t+s} \rho(0)(1-e^{-\xi u}) \, du  \leq s\,,
\end{align}
and the result follows.

\end{proof}

\begin{lemma}\label{bar_X}
Let $\bar X = \tilde X_{t+s}-\tilde X_t - \tfrac{1}{\ab A \ab} \int_A (\tilde X_{t+s} - \tilde X_t) \, dx$. Then for all $x \in A$ we have
\begin{align*}
\E [\bar X^2(x)] & \leq s + \tfrac{c}{\ab A \ab} e^{-2t}\,.
\end{align*}
\end{lemma}

\begin{proof}
We can replace $\tilde X_t$, $\tilde X_{t+s}$ by $X_t$, $X_{t+s}$ due to cancellations:
\begin{align*}
\bar X &= X_{t+s} - X_t - \frac{1}{\ab V \ab} \int_V (X_{t+s}-X_t ) \, dx \\
& \quad - \frac{1}{\ab A \ab} \int_A \big( X_{t+s} - X_t  - \frac{1}{\ab V \ab} \int _V \big( X_{t+s} - X_t \big) \, dy \big) \, dx \\
&= X_{t+s} - X_t - \frac{1}{\ab V \ab} \int_V (X_{t+s}-X_t ) \, dx \\
& \quad - \frac{1}{\ab A \ab} \int (X_{t+s} - X_t) \, dx  + \frac{1}{\ab V \ab} \int_V (X_{t+s}-X_t) \, dy \\
&= X_{t+s} - X_t - \frac{1}{\ab A \ab} \int (X_{t+s} - X_t) \, dx\,.
\end{align*}
By \eqref{eq_trivial}
\begin{align*}
\E[(X_{t+s}-X_t)^2] &\leq s\,,
\end{align*}
and by the calculations done in Lemma \ref{y_lemma} we have
\begin{align*}
\E\Big[ \big(\tfrac{1}{\ab A \ab} \int_A ( X_{t+s} -  X_t) \, dx  \big)^2 \Big] & \leq \frac{c}{\ab A \ab} e^{-2t}(1-e^{-2s})\,.
\end{align*}
Thus, since $\rho$ is non-negative, we have 
\begin{align*}
\E [\bar X(x)^2] & \leq s + \frac{c}{\ab A \ab} e^{-2t}(1-e^{-2s}) - \frac{2}{\ab A \ab} \int_A \E[(X_{t+s}(x)-X_t(x))( X_{t+s}(y) - X_t(y)]\, dy  \\
&= s + \frac{c}{\ab A \ab} e^{-2t}(1-e^{-2s}) - \frac{2}{\ab A \ab} \int_A \int_t^{t+s} \rho(e^u( x-y ))(1-e^{-\xi u}) \, du \, dy \\
& \leq  s + \frac{c}{\ab A \ab}  e^{-2t}\,.
\end{align*}
\end{proof}

\subsubsection{Concentration inequality for zero mean almost $\star$-scale invariant GMC}

In this section we recall parts of the proof of the concentration inequality of \cite{LRV} (Proposition 6.2.) for GMC built out of a zero-mean almost $\star$-scale invariant field.

Let $X$ be an almost $\star$-scale invariant field with parameter $\xi$ and $V \subset \T^2$ a Borel set with positive Lebesgue measure. We denote 
\begin{align*}
\tilde X(z) &:= X(z) - \frac{1}{\ab V \ab} \int_V X(z) \, dz \,,
\end{align*}
where we again use the notation $dx = dv_{\T^d}(x)$.

\begin{lemma}\label{concentration}
Let $t$ be large enough so that $\ab V \ab \geq e^{-2t}$. Then for any Borel set $D \subset V$ with $\ab D \ab \geq e^{-2t}$ we have 
\begin{align*}
\P\big( \tilde M_{\gamma}^{(t)}(D) \leq \tfrac{\ab D \ab}{2} \big) &\leq  \exp \big( - c \ab D \ab e^{2t} \big)\,.
\end{align*}
\end{lemma}

\begin{proof} 
We will use Kahane's convexity inequality (see for example Corollary 6.2. in \cite{RoVa}). We start by estimating the covariance kernel of $\tilde X^{t}$. We have
\begin{align*}
\E[ \tilde X^{t}(z) \tilde X^{t}(w)] &= \int_t^\infty \rho(e^u ( z-w ))(1-e^{-\xi u}) \, du - \frac{1}{\ab V \ab} \int_V \int_t^\infty \rho(e^u ( z'-w ))(1-e^{-\xi u}) \, du \, dz' \\
& \quad -  \frac{1}{\ab V \ab} \int_V \int_t^\infty \rho(e^u ( z-w' ))(1-e^{-\xi u}) \, du \, dw' \\
& \quad + \frac{1}{\ab V \ab^2} \int_{V^2}  \int_t^\infty \rho(e^u ( z'-w' ))(1-e^{-\xi u}) \, du \, dz' \, dw' 
\end{align*}
Recall that $\rho$ is non-negative. By Lemma \ref{y_lemma} we get the inequalities
\begin{align}\label{eq_tilde_X_cov}
 - c_1 e^{-2t} &\leq  \E[\tilde X^{(t)}(z) \tilde X^{(t)}(w)] - \int_t^\infty \rho(e^u ( z-w ) )(1-e^{-\xi u}) \, du \leq   c_2 e^{-2t}\,.
\end{align}
By using the left-hand side of the above inequality, we get that 
\begin{align}\label{eq_eta}
\tilde M_\gamma^{(t)}(D) &= \int_D e^{\gamma \tilde X^t(z) - \frac{\gamma^2}{2} \E[X^{t} (z)^2]} \, dz  \nonumber \\
&= \int_D e^{\gamma \tilde X^t(z) - \frac{\gamma^2}{2} \E[\tilde{X}^{t} (z)^2]} e^{\frac{\gamma^2}{2} \E[ \tilde{X}^{(t)}(z)^2 - X^{(t)}(z)^2 ]} \, dz \nonumber \\
& \geq e^{-c_1  e^{-2t} } M_{\tilde X^{(t)},\gamma}(D) \nonumber \\
&\geq \frac{1}{1+ \eta}  M_{\tilde X^{(t)},\gamma}(D)
\end{align}
for any $\eta \grea 0$ once $t$ is large enough. Now we are ready to use Kahane's inequality. Let $X_\infty^t$ be the rough part of an exactly $\star$-scale invariant field, i.e. a Gaussian field with the covariance kernel
\begin{align*}
\E[X_\infty^t(z) X_\infty^t(w)] &= \int_t^\infty \rho(e^u (z-w) ) \, du\,.
\end{align*}
We denote the corresponding GMC measure by
\begin{align*}
M_{X_\infty^t,\gamma} (dz) &:= e^{\gamma X_\infty^t(z) - \frac{\gamma^2}{2} \E[ X_\infty^t(z)^2] } \, dz \,.
\end{align*}
Let $N \sim \mathcal{N}(0,1)$ be an independent standard Gaussian. We first apply \eqref{eq_eta}, and then Kahane's convexity inequality (which is valid by \eqref{eq_tilde_X_cov}) to get
\begin{align*}
\E \big[ \exp \big(- (1+\eta) r \tilde M^{(t)}_{\gamma}(D) \big)\big] & \leq \E \big[ \exp \big(- r  M_{\tilde X^{(t)},\gamma}(D) \big)\big] \\
& \leq \E \big[ \exp \big(-r e^{\gamma \sqrt{c_2  e^{-2t}} N - \frac{\gamma^2}{2}c_2  e^{-2t}} M_{X_\infty^t,\gamma}(D) \big)\big] \\
& \leq \E \big[ \exp \big(-r  e^{-\gamma c- \frac{\gamma^2}{2}c_2  e^{-2t}} M_{X_\infty^t,\gamma}(D) \big) + \Pb(\ab N \ab \geq c (c_2  e^{-2t})^{-1/2}) \big] \\
& \leq \E \big[ \exp \big(-r   e^{-\gamma c - \frac{\gamma^2}{2} } M_{X_\infty^t,\gamma}(D) \big) + \Pb(\ab N \ab \geq c (c_2  e^{-2t})^{-1/2}) \big] \\
& \leq \E \big[ \exp \big(- \tfrac{9r }{10}  M_{X_\infty^t,\gamma}(D) \big) + \Pb(\ab N \ab \geq c (c_2  e^{-2t})^{-1/2}) \big]
\end{align*}
for a suitable $c$. The second term is bounded by $C_1 e ^{-C_2 e^{2t}}$ for some $C_1, C_2$. By Appendix C in \cite{LRV} we have 
\begin{align*}
\E  \big[ \exp \big(- r  \big( M_{X_\infty^t,\gamma}(D) - \ab D \ab \big) \big) \big] & \leq \exp \big( c_p r^p e^{2t(1-p)} \ab D \ab  \big)
\end{align*}
for all $p \in (1,4/\gamma^2)$ and $r \leq e^{2t}$. It follows that 
\begin{align*}
\E \big[ \exp \big(- (1+\eta)r  \big( \tilde M^{(t)}_{\gamma}(D) - \ab D \ab \big) \big) \big] & \leq  \E \big[ \exp \big(- \tfrac{9r}{10} M_{X_\infty^t,\gamma}(D) \big) + C_1 e^{-C_2 e^{2t}} \big]   e^{(1+\eta)r \ab D \ab }  \\
& \leq e^{\tilde c_p r^p e^{2t(1-p)} |D| + (\frac{1}{10} + \eta)r \ab D \ab     } + C_1 e^{(1+\eta)r \ab D \ab - C_2 e^{2t}}\,,
\end{align*}
and
\begin{align*}
\P\big( \tilde M_{\gamma}^{(t)}(D) \leq \tfrac{\ab D \ab}{2} \big) &= \P\big( -r \tilde M_{\gamma}^{(t)}(D) \geq  -r \tfrac{\ab D \ab}{2} \big) \\
&\leq e^{r \frac{\ab D \ab }{2}} \E \big[ \exp \big( -r \tilde M_{\gamma}^{(t)}(D)   \big) \big ] \\
&= e^{- r \frac{\ab D \ab }{2}} \E \big[ \exp \big( -r (\tilde M_{\gamma}^{(t)}(D)   - \ab D \ab ) \big) \big] \\
& \leq e^{- r \frac{\ab D \ab }{2}} e^{\tilde c_p r^p e^{2t(1-p)} |D| + \frac{\frac{1}{10} + \eta }{1 + \eta}  r \ab D \ab } + C_1 e^{(1+\eta)r \ab D \ab - C_2 e^{2t}}\,.
\end{align*}
Now, for $r = \delta e^{2t}$ for small enough $\delta$ we get that
\begin{align*}
\P\big( \tilde M_{\gamma}^{(t)}(D) \leq \tfrac{\ab D \ab}{2} \big) &\leq \exp \big( (- (\tfrac{1}{2} - \tfrac{\frac{1}{10} + \eta }{1 + \eta}) \delta \ab D \ab   + \tilde c_p \delta^p |D|  ) e^{2t} \big) + C_1 \exp \big( ( (1+\eta)\delta  \ab D \ab - \tfrac{C_2}{|V|} |D|) e^{2t} \big)\,,
\end{align*}
which is of the wanted form for small enough $\delta$ and $\eta$.

\end{proof}

\subsection{Lower bound}

In this section we will derive a lower bound for $\P(M_{\tilde Z,\gamma}(\tor) < \eps)$ for a rather general class of $\log$-correlated fields $Z$ on $\tor$. We will apply the Donsker--Varadhan theorem (see e.g. Proposition 2.3. in \cite{BuDu}).
	\begin{definition}
	Let $\Q$ and $\P$ be probability measures such that $\Q$ is absolutely continuous with respect to $\P$ (denoted by $\Q \ll \P$). We define the relative entropy between $\Q$ and $\P$ by
	\begin{align*}
	\opn{Ent}(\Q,\P) &:= \E_\Q[\log \frac{d\Q}{d\P} ]\,,
	\end{align*}
	where $\frac{d\Q}{d\P}$ is the Radon--Nikodym derivative.
	\end{definition}  
  
  \begin{theorem}[Donsker--Varadhan]\label{thm:donsker}
	Let $\Pb$ be a probability measure on a Polish space $E$ and $k: E \to \R$ a random variable that is bounded from below. Then 
	\begin{align*}
	- \log \E_\Pb[e^{-k}] &= \inf_{\mathbb{Q} \ll \Pb } \big( \E_{\mathbb{Q}}[k] + \opn{Ent}(\mathbb{Q},\mathbb{P}) \big)\,.
	\end{align*}
  \end{theorem}
  
  \begin{remark}
  Let $N_1 \sim \mathcal{N}(0,\sigma_1^2)$ and $N_2 \sim \mathcal{N}(0,\sigma_2^2)$ be two Gaussian random variables. The relative entropy between them is given by
  \begin{align}\label{eq:entropy_of_gauss}
  \opn{Ent}(N_1,N_2) &= \log \frac{\sigma_2}{\sigma_1} + \frac{1}{2} \big( \frac{\sigma_1^2}{\sigma_2^2} - 1 \big)\,.
  \end{align}
  
  \end{remark}
	\begin{lemma}\label{lem:dv_bound}
	Let $Z$ be a logarithmically correlated Gaussian field satisfying the conditions of Theorem \ref{thm:decomposition}. Denote $\tilde Z = Z - \int_{\T^d} Z \, dv_{\T^d}$. Then we have that 
	\begin{align}\label{eq:dv_ub}
- \log \E [ e^{- R^{d + \frac{\gamma^2}{2}} M_{\tilde Z,\gamma}(\T^d)} ] \leq c R^d
\end{align}
	\end{lemma}
	\begin{proof}
	We give the proof for $d=2$. The generalisation to $d \geq 2$ is straightforward.
	
	We have $Z = X +  H$, where $X$ has the covariance operator $\PNN ( - \Delta^{-1} + \Delta^{-1-\xi})$ and $H$ is an independent Gaussian field with $H \in H^{1+s}(\T^2)$ almost surely for some $s>0$. Let $\mu$ be the law of $X$. Now
	\begin{align}\label{eq:donsker_bound}
	- \log \E [ e^{- R^{2 + \frac{\gamma^2}{2}} M_{\tilde Z,\gamma}(\T^2)} ] &= - \log  \E [ \E_{\mu} [ e^{- R^{2 + \frac{\gamma^2}{2}} M_{\tilde X+\tilde H,\gamma}(\T^2)}  ]] \nonumber \\
	& \leq \E[ -\log \E_{\mu} [e^{- R^{2 + \frac{\gamma^2}{2}} M_{\tilde X+ \tilde H,\gamma}(\T^2)}  ] ] \nonumber \\
	&= \E[ \inf_{\nu \ll \mu} \big( R^{2+\frac{\gamma^2}{2}} \E_\nu [M_{\tilde X+\tilde H,\gamma}(\T^2)] + \opn{Ent}(\nu,\mu) \big) ]\,.
	\end{align}
	We will choose a suitable $\nu \ll \mu$ to obtain an upper bound.

  The Fourier transform of $X$ is given by 
	\begin{align*}
	  \hat X(k) &= \frac{\alpha_k}{\sqrt{2\pi} |k|} \big( 1 - \tfrac{1}{|k|^{2\xi}} \big)^{1/2}  \mathbf{1}_{|k| >N} \,, \quad k \in \Z^2\,,
\end{align*}	  
	where $(\alpha_k)_k$ is a sequence of i.i.d. standard Gaussians. We will choose $\nu$ to be the law of the Gaussian field $Y: \T^2 \to \R$, defined by
  \begin{align}\label{eq:Y_field}
	\hat Y(k) &= \begin{cases}
	\frac{|k|}{R} \hat X(k)\,, \quad 0 \leq  |k| \leq R \,, \\
	\hat X(k) \,, \quad |k| > R \,.
	\end{cases}
  \end{align}
  and we assume that $R > N$. Note that only finitely many Fourier modes of $\mu$ and $\nu$ differ, so $\nu$ is absolutely continuous with respect to $\mu$. Now it suffices to bound the quantity
  \begin{align*}
  R^{2+ \frac{\gamma^2}{2}} \E \Big[ \int_{\T^2} e^{\gamma (\tilde Y + \tilde H) - \frac{\gamma^2}{2} \E[\tilde Z^2] } \, dv_{\T^2} \Big] + \opn{Ent}(\nu,\mu)\,,
  \end{align*}
  where $Y$ and $H$ are independent. This has been done in Lemmas \ref{lem:entropy} and \ref{lem:E_nu} and the claim follows.
	\end{proof}
	
	\begin{corollary}
	There exists $c_\gamma > 0$ such that 
	\begin{align*}
	\P(M_{\tilde Z,\gamma}(\tor) < \eps) & \geq \tfrac 12 \exp \big( - c_\gamma \eps^{- \frac{4}{\gamma^2}} \big)\,.
	\end{align*}
	\end{corollary}
	\begin{proof}
	Let $a>0$. We have 
	\begin{align*}
	\E[e^{ - R^{2 + \frac{\gamma^2}{2}} M_{\tilde Z,\gamma}(\tor)}] &= \E[e^{ - R^{2 + \frac{\gamma^2}{2}} M_{\tilde Z,\gamma}(\tor)} ( \mathbf{1}_{M_{\tilde Z,\gamma} < a R^{- \frac{\gamma^2}{2}}} + \mathbf{1}_{M_{\tilde Z,\gamma}(\tor) \geq a R^{- \frac{\gamma^2}{2}} })]\\
	& \leq \P(M_{\tilde Z,\gamma}(\tor) < a R^{- \frac{\gamma^2}{2}} ) + \exp \big(- a R^{2 } \big) \,.
	\end{align*}
	By Lemma \ref{lem:dv_bound} we get
	\begin{align*}
	\P(M_{\tilde Z,\gamma}(\tor) < a R^{- \frac{\gamma^2}{2}}) & \geq \exp \big(- c R^2 \big) - \exp \big(- aR^{2 }  \big) \,.
	\end{align*}
	For $a > c$ and $\eps = aR^{- \frac{\gamma^2}{2}}$ this implies that 
	\begin{align*}
	\P(M_{\tilde Z,\gamma}(\tor) < \eps) & \geq \exp \big( - c a^{\frac{4}{\gamma^2}} \eps^{ - \frac{4}{\gamma^2}} \big)  - \exp \big( - a^{\frac{4}{\gamma^2}+1} \eps^{- \frac{4}{\gamma^2}} \big)  \geq \tfrac 12 \exp \big( - c a^{\frac{4}{\gamma^2}} \eps^{ - \frac{4}{\gamma^2}} \big)
	\end{align*}
	for $\eps$ small enough.
	\end{proof}

     \begin{lemma}\label{lem:entropy}
     We have
	\begin{align*}
  \mathrm{Ent}(\nu,\mu)\leq C R^{2}\,.
	\end{align*}
  \end{lemma}
  \begin{proof}
By independence of the Fourier modes, we can write 
\begin{align*}
\nu = \bigotimes_{k\in \mathbb{Z}^2 } \nu_{k} \qquad \mu= \bigotimes_{k \in \mathbb{Z}^2 } \mu_k
\end{align*}
where $\nu_k$ and $\mu_{k}$ are the marginals on the $k$th Fourier mode. Then the relative entropy is 
\begin{align*}
\mathrm{Ent}(\nu,\mu)=\sum_{k\in\mathbb{Z}^2 } \mathrm{Ent}(\nu_k,\mu_k).
\end{align*}
This can be computed explicitly, since $\mu_k, \nu_k$ are one-dimensional Gaussians. Since $\mu_k, \nu_k$ coincide for $|k|\grea R$, only the terms with $\ab k \ab \leq R$ are non-zero. We get
\begin{align*}
\opn{Ent}(\nu,\mu) &= \sum_{k \in \Z^2, 0 <  |k| \leq R} \Big( \log \frac{ \E [\hat X(k)^2]^{1/2}  }{\E[\hat Y(k)^2]^{1/2}} + \frac{1}{2} \big( \frac{\E[\hat Y(k)^2]}{\E[\hat X(k)]^2} -1  \big) \Big) \\
&= \sum_{0 <  |k| \leq R} \big( \log \tfrac{R}{|k|} + \tfrac{1}{2} \big( \tfrac{|k|^2}{R^2} -1 \big) \big) \\
& \leq \sum_{0 <  |k| \leq R} \log \tfrac{R}{|k|}\,.
\end{align*}
  By a relabelling we get 
  \[ \sum_{ 0 \less \ab k \ab \leq R} \log \tfrac{R}{\ab k \ab } = \sum_{k \in\frac {1}{R} \mathbb{Z}^{2}, 0 \less |k|\leq 1}  \log(\tfrac{1}{|k|}). \]
  Dividing this by $R^{2}$ we end up with 
  \[
   \sum_{k \in\frac {1}{R} \mathbb{Z}^{2}, 0 \less |k| \leq 1}   \log(\tfrac{1}{|k|}) \frac{1}{R^2}\,,
  \]
  which converges to 
  $\int_{|x|\leq 1} -\log |x| \, d^2x <\infty$  as $R\to \infty$.  
  \end{proof}
  
  \begin{lemma}\label{lem:mu_nu_bounds}
  Let $\PR$ be the projections to Fourier modes with  $\ab k \ab \leq R$. Then 
  \begin{itemize}
  \item[(i)] 
  \begin{align*}
  \E[( \PR Y(0)^2] &\leq C\,.
  \end{align*}
  \item[(ii)]
  \begin{align*}
  \E[ ( \PR X(0))^2 ] & \geq \log \tfrac R N\,.
  \end{align*}
  \end{itemize}
  \end{lemma}
  \begin{proof} \begin{itemize} \item[(i)]
  \begin{align*}
   \mathbb{E}[( \PR Y(0))^2] &= \tfrac{1}{2\pi} \sum_{k \in \Z^2, 0<|k|\leq R} \tfrac{|k|^2}{R^2} \Big(  \big( \tfrac{1}{|k|^2} - \tfrac{1}{|k|^{2+2\xi}} \big) \mathbf{1}_{|k| >N}  \Big)  \leq \tfrac{1}{2\pi} \sum_{0<|k| \leq R}  \tfrac{1}{R^2}  \leq C \,.
   \end{align*}
   \item[(ii)]
   \begin{align*}
   \E[(\PR X(0))^2] &= \tfrac{1}{2\pi} \sum_{k \in \Z^2, |k| \leq R} \Big(  \big( \tfrac{1}{|k|^2} - \tfrac{1}{|k|^{2+2\xi}} \big) \mathbf{1}_{|k| > N}  \Big)  \geq \tfrac{1}{2\pi} \int_{N<|x|<R} \tfrac{1}{|x|^2} \, d^2x = \log \tfrac{R}{N} \,.
   \end{align*}
  \end{itemize}
  
  \end{proof}
  
\begin{lemma}\label{lem:E_nu}
Let $Y$ be the Gaussian field defined in \eqref{eq:Y_field}. Then we have
\begin{align*}
\E \Big[ \int_{\T^2} e^{\gamma (\tilde Y + \tilde H) - \frac{\gamma^2}{2} \E[\tilde Z^2]} \, dv_{\T^2} \Big] \leq C R^{-\frac{\gamma^2}{2}}\,.
\end{align*}
\end{lemma}  
\begin{proof}
We can write (note that $\tilde Y = Y$)
\begin{align*}
\int_{\T^2} e^{\gamma (\tilde Y + \tilde H) - \frac{\gamma^2}{2} \E[\tilde Z^2]} \, dv_{\T^2}  = e^{-\frac{\gamma^2}{2} \E[( \PR X(0))^2] }\int e^{\gamma  ( \PR Y + \tilde H) - \frac{\gamma^2}{2} \E[\tilde H^2]} e^{\gamma \PRR Y- \frac{\gamma^2}{2} \E[ (\PRR X )^2 ]} \, dv_{\T^2}
\end{align*}
where $\PRR$ is the projection on Fourier-modes with $ \ab k \ab \grea R$. Now by independence of $\PR Y$ and $ \PRR Y $ we have (note that the dependency on $H$ cancels out)
\begin{align*}
\E \Big[ \int_{\T^2} e^{\gamma (\tilde Y + \tilde H) - \frac{\gamma^2}{2} \E[\tilde Z^2]} \, dv_{\T^2} \Big]  &= e^{-\frac{\gamma^2}{2} \E[ (\PR X(0))^2 ]} \int  \mathbb{E} [\exp(\gamma \PR Y) ] \E e^{ \gamma \PRR Y - \frac{\gamma^2}{2} \E[(\PRR X)^2} \, dv_{\T^2}   \\
&=  e^{-\frac{\gamma^2}{2} \E[ ( \PR X)^2 ]} \int e^{\frac{\gamma^2}{2} \E[ ( \PR Y)^2 ] + \frac{\gamma^2}{2} \E[ ( \PRR Y )^2] - \frac{\gamma^2}{2} \E[ ( \PRR X)^2]} \, dv_{\T^2} \\
&= e^{-\frac{\gamma^2}{2} \E[ (\PR X(0))^2 ]} \int e^{\frac{\gamma^2}{2} \E[ ( \PR Y)^2 ]} \, dx \\
& \leq e^{-\frac{\gamma^2}{2} \log R + \frac{\gamma^2}{2} C} \\
&= e^{\frac{\gamma^2}{2} C} R^{-\frac{\gamma^2}{2}}\,,
\end{align*}
where we used Lemma \ref{lem:mu_nu_bounds} to obtain the upper bound.
\end{proof}

\section{Massless Sinh--Gordon model on a torus}\label{shg_section}

As an application of our small deviations bounds for GMC, we derive upper and lower bounds for the free energy of the massless Sinh--Gordon model on the two-dimensional torus.

\subsection{Gaussian Free Field on the two-dimensional torus}\label{s_gff}

Let $R \grea 0$ and $\toR$ be the $R$-torus
\begin{align*}
\toR &:= \R^2 / (R \Z)^2\,.
\end{align*}
We denote by $g_R$ the flat metric on $\toR$. Let $\psi_R: \toR \to \tor$ be given by $\psi_R(z) = \frac{z}{R}$. Then it holds that
\begin{align}\label{g_scaling}
g_R &= R^2 \psi_R^* g_1\,,
\end{align}
where $\psi_R^*$ denotes the pullback. Let $\tilde X_R: \toR \to \R$ be the zero-mean Gaussian free field on $\toR$. Its covariance kernel is given by the zero-mean Green function 
\begin{align}\label{green}
\E[\tilde X_R(z) \tilde X_R(w)] &= G_R(z,w) = -\log d_{\toR}(z,w) + h_R(z,w)\,,
\end{align}
where $d_R$ is the distance function corresponding to $g_R$ and $h_R$ is a smooth function given by
\begin{align*}
h_R(z,w) &=  m_R( \log d_\toR(z,\cdot)) + m_R(\log d_\toR(\cdot,w)) - m_R(\log d_\toR(\cdot, \cdot) )\,,
\end{align*}
where $m_R$ denotes taking the average over $\toR$.

The scaling property \eqref{g_scaling} implies that
\begin{align}\label{gff_scaling}
\tilde X_R(z) \overset{d}{=} \tilde X_1(\tfrac{z}{R})\,,
\end{align}
which is equivalent to the scaling property of the covariance kernel
\begin{align*}
G_R(z,w) = G_1(\tfrac zR, \tfrac wR)\,.
\end{align*}

We denote the corresponding GMC by
\begin{align}\label{gmc}
M_{\gamma,R}(dz) 
&:= \lim_{\eps \to 0}  e^{\gamma \tilde X_{R,\eps}(z) - \frac{\gamma^2}{2} s_{R,\eps}(z)}  d^2z \,,
\end{align}
where
\begin{align*}
\tilde X_{R,\eps}(z) &:= \tfrac{1}{2\pi} \int_0^{2\pi} \tilde X_R(z+\eps e^{i \theta}) \, d \theta\,, \\
s_{R,\eps}(z) &:= \E[\tilde X_{R,\eps}(z) \tilde X_{R,\eps}(z)] - h_R(z,z)\,.
\end{align*}
The limit \eqref{gmc} exists weakly in probability, see \cite{Ber}. By using \eqref{gff_scaling} we get the GMC scaling relation
\begin{align}\label{gmc_scaling}
\int_{\toR} M_{\gamma,R}(dz) 
&= R^{2 + \frac{\gamma^2}{2}} \int_{\tor} M_{\gamma,1}(dz)\,,
\end{align}
see for example Proposition 2.2 in \cite{KuOi20} for a detailed proof.

\subsection{Path integral formulation of the massless Sinh--Gordon model}

\noindent The classical action functional of the model on $(\toR, g_R)$ given by
\begin{align*}
S(\varphi) &= \int_{\toR} \big( \tfrac 12 \ab d \varphi \ab_R^2 + \tfrac{Q}{4\pi} K_R \varphi + \mu e^{\gamma \varphi} + \mu e^{-\gamma \varphi} \big) dv_{R}\,,
\end{align*}
where $v_R$ is the volume form, $K_R$ the scalar curvature and $\ab \cdot \ab_R$ the norm induced by the Riemannian metric $g_R$. As $g_R$ is flat, we have $K_R \equiv 0$, and we get the coordinate expression
\begin{align}\label{action_coordinates}
S(\varphi) &= \int_{[0,R]^2} \big( \tfrac 12 \ab \nabla \varphi(z) \ab^2 + \mu e^{\gamma \varphi(z)} + \mu e^{-\gamma \varphi(z)} \big) \, dz \,.
\end{align}
\begin{remark}
For what follows we could just as well use the more general action
\begin{align*}
\int_{[0,R]^2} \big( \tfrac 12 \ab \nabla \varphi(z) \ab^2 + \mu_1 e^{\gamma_1 \varphi(z)} + \mu_2 e^{-\gamma_2 \varphi(z)} \big) \, dz\,,
\end{align*}
$\mu_i > 0$, $\gamma_i \in (0,2)$, but we will work with \eqref{action_coordinates} to keep the notation simple.
\end{remark}
\noindent The massless free field on $\toR$ is given by $\varphi(z) = c+\tilde X_R(z)$, where $\tilde X_R$ is the zero mean GFF on $\toR$ and $c \in \R$ is distributed according to the Lebesgue measure. Let us now consider the partition function of the massless Sinh--Gordon model, which we obtain by tilting the law of the massless free field by the GMC terms
\begin{align*}
Z_R &= \E_R \Big[ \int_\R  e^{ - \mu e^{\gamma c} M_{\gamma,R}(\toR) - \mu e^{-\gamma c} M_{-\gamma,R}(\toR)  } \, dc \Big] \,,
\end{align*}
where $\E_R$ is the expectation with respect to $\tilde X_R$. By the GMC scaling \eqref{gmc_scaling}, we get that 
\begin{align*}
Z_R &= \E_1 \Big[ \int_\R e^{- \mu R^{2 + \frac{\gamma^2}{2}} (e^{\gamma c} M_{\gamma,1}(\tor) + e^{-\gamma c} M_{-\gamma,1}(\tor)) } \, dc \Big]\,.
\end{align*}
By a change of variables $e^{\gamma c} \sqrt{M_{\gamma,1}(\tor)} \to  \sqrt{ M_{-\gamma,1}(\tor)}c' $ we get that
\begin{align*}
Z_R &= \tfrac 1 \gamma \E_1 \Big[ \int_0^\infty e^{- \mu R^{2+\frac{\gamma^2}{2}} \sqrt{M_{\gamma,1}(\tor) M_{-\gamma,1}(\tor)} (c+c^{-1}) } \tfrac{dc}{c} \Big]\,.
\end{align*}
We want to derive upper and lower bounds for (the logarithm of) this quantity.

\begin{remark}
We also have the identity
\begin{align}\label{eq_bessel}
Z_R &= \tfrac 1 \gamma \E_1 \Big[ K_0 \big( 2\mu  R^{2 + \frac{\gamma^2}{2}} \sqrt{M_{\gamma,1}(\tor) M_{-\gamma,1}(\tor)} \big) \Big]\,,
\end{align}
where $K_0$ is the modified Bessel function of the second kind with index $0$.
\end{remark}

\subsection{Free energy bounds: proof of Theorem \ref{fe_bounds}}

In this section we give a proof of Theorem \ref{fe_bounds}.

\subsubsection{Lower bound}

\begin{proposition}\label{prop_Z_ub}
 For large enough $R$ we have the bound
\begin{align}\label{log_Z_lb}
- \log Z_R & \geq f_\gamma \mu^{\frac{2}{\gamma Q}} R^2
\end{align}
for some $f_\gamma > 0$.
\end{proposition}

\begin{proof}
We write the proof for $\mu=1$, as the result for general $\mu>0$ follows by replacing $R$ by $\mu^{\frac{1}{\gamma Q}} R$.

We denote $\mathcal{M} := \sqrt{M_{\gamma,1}(\tor)M_{-\gamma,1}(\tor)}$. First we split the expectation into two parts using the events $\{\mathcal{M} \leq   R^{- \frac{\gamma^2}{2}}\}$ and $\{ \mathcal{M} \grea  R^{- \frac{\gamma^2}{2}}\}$. We estimate the two terms separately. First, for the latter we have
\begin{align*}
\E_1 \Big[ \mathbf{1}_{ \{   \mathcal{M} \grea  R^{-\frac{\gamma^2}{2}} \} } \int_0^\infty e^{-  R^{2 + \frac{\gamma^2}{2}}  (c+c^{-1}) \mathcal{M} } \tfrac{1}{c} \, dc \Big] & \leq \E_1 \Big[ \mathbf{1}_{ \{  \mathcal{M} \grea   R^{-\frac{\gamma^2}{2}} \} } \int_0^\infty e^{-  R^2  (c+c^{-1}) } \tfrac{1}{c} \, dc\Big] \\
&\leq \int_0^\infty e^{-   R^2  (c+c^{-1}) } \tfrac{1}{c} \, dc\,.
\end{align*}
By Laplace's method, this integral behaves as $C  R^{-1} e^{-2  R^2}$ when $R$ is large, so for sufficiently large $R$ we have 
\begin{align*}
\int_0^\infty e^{-   R^2  (c+c^{-1}) } \tfrac{1}{c} \, dc & \leq e^{-2   R^2}\,.
\end{align*}
Next, we deal with the event $\{ \mathcal{M} \leq  R^{-\frac{\gamma^2}{2}} \}$. We split this further with the event $\{   \mathcal{M} \grea  R^{-2 - \frac{\gamma^2}{2} + \delta } \}$ where $\delta \in (0, \frac{\gamma^2}{2})$. We have
\begin{align*}
\E_1 \Big[ \mathbf{1}_{ \{ R^{-2 - \frac{\gamma^2}{2} + \delta } \les  \mathcal{M} \leq  R^{-\frac{\gamma^2}{2}} \} } \int_0^\infty e^{-  R^{2 + \frac{\gamma^2}{2}}  (c+c^{-1}) \mathcal{M} } \tfrac{dc}{c} \Big] & \leq \P \big(   \mathcal{M} \leq  R^{- \frac{\gamma^2}{2}} \big) \int_0^\infty e^{-  R^{\delta} (c+c^{-1}) } \tfrac{dc}{c}  \\
& \leq \P \big(   \mathcal{M} \leq  R^{- \frac{\gamma^2}{2}} \big) e^{-2  R^\delta} \\
& \leq \P \big(   \mathcal{M} \leq  R^{- \frac{\gamma^2}{2}} \big)\,,
\end{align*}
where the second inequality follows by applying Laplace's method.

For the last remaining term, we estimate
\begin{align*}
\E_1 \Big[ \mathbf{1}_{ \{  \mathcal{M} \leq  R^{-2 - \frac{\gamma^2}{2} + \delta }  \} } \int_0^\infty e^{- R^{2 + \frac{\gamma^2}{2}}  (c+c^{-1}) \mathcal{M} } \tfrac{1}{c} \, dc \Big] & \leq \E_1 \Big[ \mathbf{1}_{ \{  \mathcal{M} \leq R^{-2 - \frac{\gamma^2}{2} + \delta }  \} } \int_0^\infty e^{-    (c+c^{-1}) \mathcal{M} } \tfrac{1}{c} \, dc \Big] \\
&=  \E_1 \Big[ \mathbf{1}_{ \{   \mathcal{M} \leq R^{-2 - \frac{\gamma^2}{2}+ \delta}  \} } K_0 ( 2 \mathcal{M} ) \Big]\,,
\end{align*}
where in the equality we used the identity \eqref{eq_bessel}. By using the fact that $K_0$ diverges logarithmically near the origin, we can estimate 
\begin{align*}
\E_1 \Big[ \mathbf{1}_{ \{   \mathcal{M} \leq  R^{-2 - \frac{\gamma^2}{2}+ \delta}  \} } K_0 ( 2  \mathcal{M} ) \Big] & \leq C \E_1 \Big[ \mathbf{1}_{   \{  \mathcal{M} \leq R^{-2 - \frac{\gamma^2}{2}+ \delta} \} }  \ab \log \big( 2  \mathcal{M} \big) \ab  \Big] \\
&\leq  C \P (  \mathcal{M} \leq R^{-2 - \frac{\gamma^2}{2}+ \delta} )^{\frac{1}{p}} \E_1 [(\log \ab 2 \mathcal{M} \ab )^q]^{\frac{1}{q}}\,,
\end{align*}
where $p,q \grea 1$ satisfy $\frac{1}{q} + \frac{1}{p} = 1$. The expected value in the above expression is finite by simple arguments using the facts that $M_{\gamma,1}(\tor)$ and $M_{-\gamma,1}(\tor)$ have finite negative moments and some finite positive moments. Now, we have shown that
\begin{align}\label{Z_ub}
Z_R & \leq  e^{-2  R^2}   + \P \big(  \mathcal{M} \leq  R^{- \frac{\gamma^2}{2}} \big)    + C \P ( \mathcal{M} \leq  R^{-2 - \frac{\gamma^2}{2}+ \delta} )^{\frac{1}{p}}\,.
\end{align}
Note that we have the union bound and $\gamma \to -\gamma$ symmetry
\begin{align}\label{eq_union_bound}
\Pb \big(\sqrt{M_{\gamma,1}(\tor) M_{-\gamma,1}(\tor)} \less \eps \big) & \leq \Pb \big( \{ M_{\gamma,1}(\tor) \less \eps \} \cup \{ M_{-\gamma,1}(\tor) \less \eps \} \big) \nonumber \\
& \leq 2 \Pb \big( M_{\gamma,1}(\tor) \less \eps \big)\,.
\end{align}
Now by Lemma \ref{lem:product_gmc} we get that 
\begin{align}\label{Z_ub2}
Z_R & \leq e^{-2R^2} + 2 e^{-c_\gamma R^2} + e^{-c R^k}
\end{align}
for some $k>2$. Now the claim follows.

\end{proof}

\begin{lemma}\label{lem:product_gmc}
We have
	\begin{align*}
	\Pb \big( M_{\gamma,1}(\tor) \less \eps \big) & \leq \exp \big(- c \eps^{- \frac{4}{\gamma^2}} \big)
	\end{align*}
	for some $c>0$ and $\eps$ small enough.
\end{lemma}
	\begin{proof} We check that $\tilde X_1$ satisfies the conditions of Theorem \ref{sd_ub}. As the constant functions are the only harmonic functions on $\tor$, it follows that for any non-constant $f \in L^2(\tor)$ we have
	\begin{align*}
	\int_{(\tor)^2} G_1(x,y) f(x) f(y) \,dv_1(x) \, dv_1(x) \grea 0\,,
	\end{align*}	
	where $G_1$ is the zero-mean Green function. Thus, if $N$ is a Gaussian random variable independent of $\tilde X_1$, then $Y = \tilde X_1 + N$ is a non-degenerate $\log$-correlated Gaussian field and it holds that $\tilde Y = \tilde X_1$ where
	\begin{align*}
	\tilde Y &:= Y  - \tfrac{1}{v_\tor(\tor)}\int_\tor Y \, dv_\tor\,.
	\end{align*}
	Now we have
	\begin{align*}
	M_{\gamma,1}(\tor) &= \int_{\tor} e^{\gamma \tilde X_1 - \frac{\gamma^2}{2} s_1} \, d^2z = \int_{\tor} e^{\gamma \tilde X_1(z) - \frac{\gamma^2}{2} \E[ (\tilde X_1(z))^2] + \frac{\gamma^2}{2} h_1(z,z)} \, d^2z  \,,
	\end{align*}
	where $h_1$ is a smooth function. Thus $M_{\gamma,1}(\tor) \geq c M_{\tilde X_1, \gamma}(\tor)$. Now note that 
	\begin{align*}
	M_{\tilde X_1,\gamma}(\tor) &= \int_\tor e^{\gamma \tilde Y - \frac{\gamma^2}{2} \E[Y^2]} e^{ \frac{\gamma^2}{2} \E[ Y^2 - (\tilde X_1)^2 ]  } \, d^2z  = e^{ \frac{\gamma^2}{2} \E[N^2]} \tilde M_{Y,\gamma}(\tor)\,. 
	\end{align*}
	Now the bound given by Theorem \ref{sd_ub} then implies
	\begin{align*}
	\P \big( M_{\gamma,1}(\tor) < \eps \big) \leq \exp \big( - c_\gamma c^{ \frac{4}{\gamma^2}} \eps^{- \frac{4}{\gamma^2}} \big)
	\end{align*}
	for some $c>0$.
	
	\end{proof}

\subsubsection{Upper bound}

	\begin{lemma}
	We have
  \begin{align}\label{eq_log-ub}
  -\log Z_R \leq \tilde f_\gamma \mu^{\frac{2}{\gamma Q}}  R^2\,,
  \end{align}
  for some $\tilde f_\gamma > 0$.
	\end{lemma} 
	\begin{proof}
	We set $\mu=1$, as replacing $R$ by $\mu^{\frac{1}{\gamma Q}}R$ will give the result for general $\mu > 0$.
	
  To begin, we truncate the zero-mode integral
  \begin{align*}
 Z_R &= \int_\R \mathbb{E}[\exp(-  R^{2+\frac{\gamma^{2}}{2}}( e^{\gamma c} M_{\gamma,1}(\tor)+ e^{-\gamma c} M_{-\gamma,1}(\tor)))] \, dc \\
  &  \geq \int_{-\frac 12}^{\frac 12}  \mathbb{E}[\exp(- R^{2+\frac{\gamma^{2}}{2}}( e^{\gamma c} M_{\gamma,1}(\tor)+ e^{-\gamma c} M_{-\gamma,1}(\tor)))] \, dc \\
  &  \geq \mathbb{E}[\exp(- e^{\frac{\gamma}{2}} R^{2+\frac{\gamma^{2}}{2}}(  M_{\gamma,1}(\tor)+  M_{-\gamma,1}(\tor)))]\,.
  \end{align*}
  By the Donsker--Varadhan theorem \ref{thm:donsker}, we have
  \begin{align*}
  &- \log \mathbb{E}[\exp(- e^{\frac{\gamma}{2}} R^{2+\frac{\gamma^{2}}{2}}(  M_{\gamma,1}(\tor)+  M_{-\gamma,1}(\tor)))] \\
  &= \inf_{\Q \ll \P} \Big( e^{\frac{\gamma}{2}} R^{2 + \frac{\gamma^2}{2}} \E_\Q[M_{\gamma,1}(\tor) + M_{-\gamma,1}(\tor)] + \opn{Ent}(\Q,\P)  \Big)\,,
  \end{align*}
  where $\P$ is the probability law of $\tilde X_1$. Now we can follow the argument given in the proof of Lemma \ref{lem:dv_bound} and the symmetry of $M_{\gamma,1}(\tor)$ under $\gamma \to -\gamma$ when $\Q$ is a law of a Gaussian field.
\end{proof}


\begin{thebibliography}{ABCD99}

 \bibitem[AH74]{AHK}Albeverio S., Høegh-Krohn, R. The Wightman axioms and the mass gap for strong interactions of exponential type in two dimensional space-time. Journal of Functional Analysis Volume 16, pp. 39-82 (1974)

\bibitem[AJJ22]{AJJ} Aru J., Jego A., Junnila J. Density of imaginary multiplicative chaos via Malliavin calculus. Probab. Theory Relat. Fields 184, 749–803 (2022).

\bibitem[BV21]{BV} Barashkov N., De Vecchi F. Elliptic stochastic quantization of Sinh-Gordon QFT, 	arXiv:2108.12664 (2021).

\bibitem[BuDu]{BuDu} Budhiraja A., Dupuis P. Analysis and Approximation of Rare Events, Springer US (2019).


\bibitem[Ber17]{Ber}
Berestycki N. An elementary approach to Gaussian multiplicative chaos. Electron. Commun. Probab. \textbf{27}, 1-12 (2017).

\bibitem[BePo]{BePo}
Berestycki N., Powell E. Gaussian Free Field and Liouville Quantum Gravity. Cambridge University Press (2025).

\bibitem[BLC22]{BLC} Bernard D., LeClair A. The sinh-Gordon model beyond the self dual point and the freezing transition in disordered systems. J. High Energ. Phys. 2022, 22 (2022).

\bibitem[CRV23]{toda} Cercl\'{e} B., Rhodes R., Vargas V. Probabilistic construction of Toda Conformal Field Theories, Annales Henri Lebesgue, Volume 6 (2023), pp. 31-64.

\bibitem[Cha25]{chatterjee} Chatterjee S. Rigorous results for timelike Liouville field theory, arXiv:2504.02348 (2025).

\bibitem[ChWi24]{witten} Chatterjee S., Witten E. Liouville Theory: An Introduction to Rigorous Approaches, arXiv:2404.02001 (2024).

\bibitem[DKRV16]{DKRV}
David F., Kupiainen A., Rhodes R., Vargas V. Liouville Quantum Gravity on the Riemann sphere. Commun. Math. Phys. \textbf{342}, 869-907 (2016).

\bibitem[DuSh11]{ds} Duplantier B., Sheffield S. Liouville quantum gravity and KPZ. Invent. math. 185, 333–393 (2011).

\bibitem[Dor98]{Dor} Dorey P. Exact S-matrices, arXiv:hep-th/9810026 (1998).


\bibitem[FLZZ98]{FLZZ} Fateev V., Lukyanov S., Zamolodchikov A., Zamolodchikov A. B. Expectation values of local fields in the Bullough-Dodd model and integrable perturbed conformal field theories, Nuclear Physics B, Volume 516, Issue 3, 1998, Pages 652-674.

\bibitem[FMS93]{FMS} Fring A., Mussardo G., Simonetti P. Form-factors for integrable Lagrangian field theories, the sinh-Gordon theory.  Nucl.Phys.B 393 (1993), 413-441.

\bibitem[FrPa77]{fro} Fröhlich J., Park Y. M. Remarks on Exponential Interactions and the Quantum Sine-Gordon Equation in Two Space-Time Dimensions, Helv. Phys. Acta 50 (1977) 315-329.

\bibitem[GGV24]{GGV} Gunaratnam T., Guillarmou C., Vargas V. 2d Sinh-Gordon model on the infinite cylinder, arXiv:2405.04076 (2024).

\bibitem[GHSS18]{GHSS} Garban C., Holden N., Sepúlveda A., Sun X. Negative moments for Gaussian multiplicative chaos on fractal sets. Electron. Commun. Probab. 23, 1-10, (2018).

\bibitem[GKR23]{imaginary} Guillarmou C., Kupiainen A., Rhodes R. Compactified Imaginary Liouville Theory, arXiv:2310.18226 (2023).


\bibitem[GKR24]{GKR} Guillarmou C., Kupiainen A., Rhodes R. Review on the probabilistic construction and Conformal bootstrap in Liouville Theory, arXiv:2403.12780 (2024).

\bibitem[GKR25]{wzw} Guillarmou C., Kupiainen A., Rhodes R. Probabilistic construction of the $\mathbb{H}^3$-Wess-Zumino-Witten conformal field theory and correspondence with Liouville theory, arXiv:2502.16341 (2025).

\bibitem[HZ24]{HZ} Hofstetter M., Zeitouni O. Decay of correlations for the massless hierarchical Liouville model in infinite volume, arXiv:2408.16649 (2024).

\bibitem[Kah85]{Kah}
Kahane J.-P. Sur le chaos multiplicatif,
  \emph{Ann. Sci. Math. Qu{\'e}bec}, \textbf{9} no.2 (1985), 105-150.
  
\bibitem[JSW19]{JSW} Junnila J., Saksman E., Webb C. Decompositions of log-correlated fields with applications.  Ann. Appl. Probab. 29(6), 3786-3820 (2019).

\bibitem[JSW20]{JSW2} Junnila J., Saksman E., Webb C. Imaginary multiplicative chaos: Moments, regularity and connections to the Ising model. Ann. Appl. Probab. 30(5), 2099-2164 (2020).

\bibitem[KLM21]{KLM} Konik R., Lájer M., Mussardo G. Approaching the self-dual point of the sinh-Gordon model. J. High Energ. Phys. 2021, 14 (2021).

\bibitem[KoMu93]{KM} Koubek A., Mussardo G. On the operator content of the sinh-Gordon model. Phys.Lett.B 311 (1993) 193-201.

\bibitem[Koz21]{koz1} Kozlowski K. Bootstrap approach to 1+1-dimensional integrable quantum field theories: the case of the Sinh-Gordon model. Proceedings of ICM 2022 PP. 4096–4118, arXiv:2112.14274 (2021).

\bibitem[Koz23]{koz2} Kozlowski K. On convergence of form factor expansions in the infinite volume quantum Sinh-Gordon model in 1+1 dimensions. Invent. math. 233, 725–827 (2023).

\bibitem[KuOi20]{KuOi20} Kupiainen A., Oikarinen, J. Stress-Energy in Liouville Conformal Field Theory. J Stat Phys 180, 1128–1166 (2020).

\bibitem[LRV22]{LRV} Lacoin H., Rhodes R., Vargas V. Path integral for quantum Mabuchi K-energy. Duke Math. J. 171 (3) 483 - 545, (2022).

\bibitem[LSSW16]{LSSW} Lodhia A., Sheffield S., Sun X., Watson S. Fractional Gaussian fields: A survey. Probab.Surv. 13 (2016) 1-56

\bibitem[Mus]{mussardo} Mussardo G. Statistical Field Theory: An Introduction to Exactly Solved Models in Statistical Physics, 2nd edn (Oxford, 2020; online edn, Oxford Academic, 21 May 2020).


\bibitem[Nik13]{Nik} Nikula M. Small deviations in lognormal Mandelbrot cascades, Electron. Commun. Probab. 25, 1-12, (2020).

\bibitem[Rem20]{remy1}
Remy G. The Fyodorov-Bouchaud formula and Liouville Conformal Field theory, \emph{Duke Mathematical Journal} {\textbf{169}} (1), 177-211 (2020).

\bibitem[ReZh20]{remy2} Remy G., Zhu T. The distribution of Gaussian multiplicative chaos on the unit interval, Annals Probab. 48 (2020) 2, 872-915.

\bibitem[RoVa10]{RoVa} Robert R., Vargas V. Gaussian multiplicative chaos revisited, Ann. Probab. 38(2): 605-631 (2010).

\bibitem[RoVa14]{review} 
Rhodes R., Vargas V. Gaussian multiplicative chaos and applications: a review, Probab. Surveys, Volume 11 (2014), 315-392.

\bibitem[RoVa19]{ldp} Rhodes R., Vargas V. The tail expansion of Gaussian multiplicative chaos and the Liouville reflection coefficient, Annals Probab. 47 (2019) 5, 3082-3107.

\bibitem[TaWi24]{TaWi} Talarczyk A., Wiśniewolski M. On small deviations of Gaussian multiplicative chaos with a strictly logarithmic covariance on Euclidean ball, 	arXiv:2401.07695 (2024).


\bibitem[Tes08]{teschner} Teschner J. On the spectrum of the Sinh-Gordon model in finite volume, Nucl.Phys.B 799 (2008) 403-429.

\bibitem[Til22]{til} Tilloy A. A study of the quantum Sinh-Gordon model with relativistic continuous matrix product states. arXiv:2209.05341 (2022).

\bibitem[UGRS25]{imaginary2} Usciati R., Guillarmou C., Rhodes R., Santachiara R. Probabilistic construction of non compactified imaginary Liouville field theory, arXiv:2505.09390 (2025).

\bibitem[Won19]{modick2} Wong M.D. Tail universality of critical Gaussian multiplicative chaos, arXiv:1912.02755 (2019).

\bibitem[Won20]{modick1} Wong M.D. Universal tail profile of Gaussian multiplicative chaos. Probab. Theory Relat. Fields 177, 711–746 (2020).


\bibitem[Zam95]{zam} Zamolodchikov Al. B. Mass scale in the sine-Gordon model and its reductions, Int.J.Mod.Phys.A 10 (1995) 1125-1150.

\bibitem[Zam06]{ztba} Zamolodchikov Al. B. On the thermodynamic Bethe ansatz equation in sinh-Gordon model, J.Phys.A 39 (2006) 12863-12887.


\end{thebibliography}
\end{document}